\documentclass[a4paper,UKenglish]{lipics}
 
\usepackage{amsmath}
\usepackage{amssymb}
\usepackage{enumerate}
\usepackage{graphicx}
\usepackage{hyperref}
\usepackage{tikz}
\tikzstyle{vertex}=[circle, draw, inner sep=0pt, minimum size=4pt]
\tikzstyle{tvertex}=[circle, draw, inner sep=0pt, minimum size=2.5pt]
\newcommand{\vertex}{\node[vertex]}
\newcommand{\tvertex}{\node[tvertex]}
\usepackage{microtype}

\newcommand{\ket}[1]{| #1 \rangle}
\newcommand{\braket}[2]{\langle #1 | #2 \rangle}

\newcommand{\bb}[1]{\mathbb{#1}}
\newcommand{\cl}[1]{\mathcal{#1}}

\title{The Minimum Size of Qubit Unextendible Product Bases}

\author{Nathaniel Johnston}
\affil{Institute for Quantum Computing, University of Waterloo\\
  Waterloo, Ontario N2L~3G1, Canada\\
  \texttt{nathaniel.johnston@uwaterloo.ca}}
\authorrunning{N. Johnston} 


\begin{document}

\maketitle

\begin{abstract}
	We investigate the problem of constructing unextendible product bases in the qubit case -- that is, when each local dimension equals $2$. The cardinality of the smallest unextendible product basis is known in all qubit cases except when the number of parties is a multiple of $4$ greater than $4$ itself. We construct small unextendible product bases in all of the remaining open cases, and we use graph theory techniques to produce a computer-assisted proof that our constructions are indeed the smallest possible.
\end{abstract}

\section{Introduction}

Unextendible product bases play a rather diverse and important role in quantum information theory \cite{DMSST03}. While their original motivation was for the construction of bound entangled states \cite{BDMSST99,LSM11,Sko11b}, they have also been used to build indecomposible positive maps \cite{Ter01}, to demonstrate Bell inequalities without a quantum violation \cite{ASHKLA11}, and demonstrate the existence of nonlocality without entanglement \cite{BDFMRSSW99}.

Furthermore, in the qubit case (i.e., the case where each local space has dimension $2$), it has been shown that unextendible product bases can be used to construct tight Bell inequalities with no quantum violation \cite{AFKKPLA12} and subspaces of small dimension that are locally indistinguishable \cite{DXY10}. It is the qubit case that we focus on in the present paper. In particular, we consider the question of how small a qubit unextendible product basis can be.

The minimum cardinality of a qubit unextendible product basis on $p$ qubits is well-known to equal $p+1$ when $p$ is odd \cite{AL01}. When $p$ is even, however, the problem is more difficult. It was shown in \cite{Fen06} that the minimum cardinality equals $p+2$ when $p = 4$ or $p \equiv 2 \, (\text{mod } 4)$. Our contribution is to solve the remaining cases (i.e., when $p \geq 8$ and $p \equiv 0 \, (\text{mod } 4)$ -- more specifically, we show that the minimum cardinality is $p+3$ when $p = 8$ and $p+4$ in all other cases.

Our approach is as follows: we formally introduce the mathematical preliminaries and graph theory techniques that we make use of in Section~\ref{sec:prelims}. We construct unextendible product bases of the claimed cardinality in Section~\ref{sec:construction}. Finally, Section~\ref{sec:minimality_proof} is devoted to the proof that there does not exist a smaller unextendible product basis in these cases.

\section{Unextendible Product Bases and Orthogonality Graphs}\label{sec:prelims}

A pure quantum state is represented by a unit vector $\ket{v} \in \bb{C}^{d_1} \otimes \cdots \otimes \bb{C}^{d_p}$ (and in our setting, $d_1 = \cdots = d_p = 2$ always). We say that $\ket{v}$ is a \emph{product state} if we can write it in the form
\begin{align*}
	\ket{v} = \ket{v_1} \otimes \cdots \otimes \ket{v_p} \ \ \text{ with } \ \ \ket{v_j} \in \bb{C}^{2} \ \forall \, j.
\end{align*}

An \emph{unextendible product basis (UPB)} is an orthonormal set $\cl{S} \subseteq (\bb{C}^2)^{\otimes p}$ of product states such that there is no product state orthogonal to every member of $\cl{S}$. It is clear that every UPB in $(\bb{C}^2)^{\otimes p}$ contains at least $p+1$ states -- if it contained only $p$ product states $\ket{v_0},\ldots,\ket{v_{p-1}}$ then we could construct another product state that is, for each $0 \leq j < p$, orthogonal to $\ket{v_j}$ on the $(j+1)$-th party and thus violate unextendibility.

It turns out that the trivial lower bound of $p+1$ states can be attained when $p$ is odd, and can almost be attained when $p$ is even, as indicated by our main result:
\begin{theorem}\label{thm:main}
	Let $f(p)$ be the smallest possible number of states in a UPB in $(\bb{C}^2)^{\otimes p}$. Then:
	\begin{enumerate}[(a)]
		\item if $p$ is odd then $f(p) = p + 1$;
		\item if $p = 4$ or $p \equiv 2 \, (\text{mod } 4)$ then $f(p) = p + 2$;
		\item if $p = 8$ then $f(p) = p + 3$;
		\item otherwise, $f(p) = p + 4$.
	\end{enumerate}
\end{theorem}

Case~(a) of Theorem~\ref{thm:main} is demonstrated by the ``GenShifts'' UPB constructed in \cite{DMSST03}. Case~(b) of Theorem~\ref{thm:main} was proved in \cite{Fen06}, and in general our techniques and presentation are similar to those of that paper. Our contribution is to prove cases~(c) and~(d) and hence complete the characterization. It is worth pointing out that cases~(c) and~(d) of Theorem~\ref{thm:main} are the first known cases (qubit or otherwise) where the minimum cardinality of a UPB exceeds the trivial lower bound $1 + \sum_{j}(d_j-1)$ by more than $1$ (see \cite{CJ13,Fen06} for several examples where the trivial lower bound is exceeded by exactly $1$).

Orthogonality graphs provide a very useful tool when dealing with unextendible product bases, particularly in the qubit case. Given a set of product states $\cl{S} = \{\ket{v_0},\ldots,\ket{v_{s-1}}\} \subseteq (\bb{C}^2)^{\otimes p}$ with $|\cl{S}| = s$, we say that the \emph{orthogonality graph of $\cl{S}$} is the graph on $s$ vertices $V := \{v_0,\ldots,v_{s-1}\}$ such that there is an edge $(v_i,v_j)$ of color $\ell$ if and only if $\ket{v_i}$ and $\ket{v_j}$ are orthogonal to each other on party $\ell$. Rather than actually using $p$ colors to color the edges of the orthogonality graph, for ease of visualization we instead draw $p$ different graphs on the same set of vertices -- one for each party (see Figure~\ref{fig:2dim_eg}).

\begin{figure}[htb]
	\centering
	\begin{tikzpicture}[x=1.4cm, y=1.4cm, label distance=0cm]     
		\vertex[fill] (v00) at (0,1) [label=90:$v_{0}$]{};
		\vertex[fill] (v01) at (-0.777,0.629) [label=141:$v_{1}$]{};
		\vertex[fill] (v02) at (-0.975,-0.225) [label=193:$v_{2}$]{};
		\vertex[fill] (v03) at (-0.434,-0.899) [label=244:$v_{3}$]{};
		\vertex[fill] (v04) at (0.434,-0.899) [label=296:$v_{4}$]{};
		\vertex[fill] (v05) at (0.975,-0.225) [label=347:$v_{5}$]{};
		\vertex[fill] (v06) at (0.777,0.629) [label=39:$v_{6}$]{};
				
		\vertex[fill] (v10) at (3.5,1) [label=90:$v_{0}$]{};
		\vertex[fill] (v11) at (2.723,0.629) [label=141:$v_{1}$]{};
		\vertex[fill] (v12) at (2.525,-0.225) [label=193:$v_{2}$]{};
		\vertex[fill] (v13) at (3.066,-0.899) [label=244:$v_{3}$]{};
		\vertex[fill] (v14) at (3.934,-0.899) [label=296:$v_{4}$]{};
		\vertex[fill] (v15) at (4.475,-0.225) [label=347:$v_{5}$]{};
		\vertex[fill] (v16) at (4.277,0.629) [label=39:$v_{6}$]{};

		\vertex[fill] (v20) at (7,1) [label=90:$v_{0}$]{};
		\vertex[fill] (v21) at (6.223,0.629) [label=141:$v_{1}$]{};
		\vertex[fill] (v22) at (6.025,-0.225) [label=193:$v_{2}$]{};
		\vertex[fill] (v23) at (6.566,-0.899) [label=244:$v_{3}$]{};
		\vertex[fill] (v24) at (7.434,-0.899) [label=296:$v_{4}$]{};
		\vertex[fill] (v25) at (7.975,-0.225) [label=347:$v_{5}$]{};
		\vertex[fill] (v26) at (7.777,0.629) [label=39:$v_{6}$]{};

		\path 
			(v00) edge (v03)
			(v00) edge (v04)
			(v00) edge (v05)
			(v00) edge (v06)
			(v01) edge (v03)
			(v01) edge (v04)
			(v01) edge (v05)
			(v01) edge (v06)
			(v02) edge (v03)
			(v02) edge (v04)
			(v02) edge (v05)
			(v02) edge (v06)

			(v11) edge (v10)
			(v11) edge (v12)
			(v13) edge (v15)
			(v13) edge (v16)
			(v14) edge (v15)
			(v14) edge (v16)

			(v20) edge (v22)
			(v21) edge (v23)
			(v23) edge (v24)
			(v25) edge (v26)
		;
	\end{tikzpicture}
	\caption{The orthogonality graph of a set of $7$ product states in $(\bb{C}^2)^{\otimes 3}$. This set of states is a product basis, since every edge is present in at least one of the three graphs, but it is extendible, since we can find a product state that is orthogonal to the states associated with $v_3, v_4, v_5, v_6$ on the first subsystem, $v_0, v_2$ on the second subsystem, and $v_1$ on the third subsystem.}\label{fig:2dim_eg}
\end{figure}
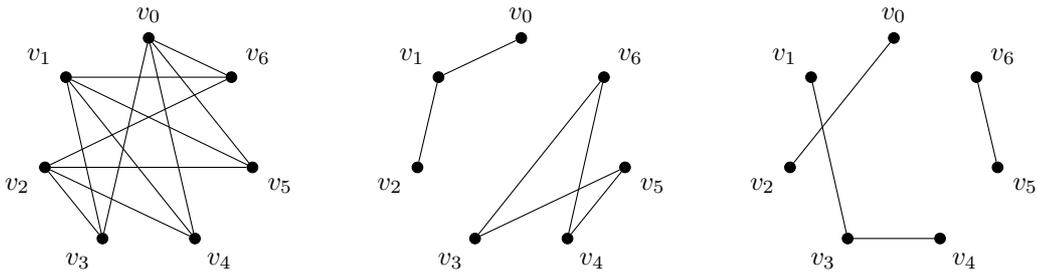

The requirement that $\cl{S}$ is an orthonormal set is equivalent to requiring that every edge is present on at least one party in its orthogonality graph. In order to help us visualize the unextendibility requirement, we make a few more observations. In particular, if $\ket{w_0},\ket{w_1},\ket{w_2} \in \bb{C}^2$ are such that $\braket{w_0}{w_1} = \braket{w_0}{w_2} = 0$, then it is necessarily the case that $\ket{w_1} = \ket{w_2}$ (up to irrelevant complex phase). It follows that the orthogonality graph associated with any qubit in a product basis is the disjoint union of complete bipartite graphs. For example, in Figure~\ref{fig:2dim_eg} the left graph is $K_{3,4}$, the center graph is the disjoint union of $K_{1,2}$ and $K_{2,2}$, and the right graph is the disjoint union of $K_{1,2}$ and two copies of $K_{1,1}$.

Furthermore, not only does every set of product states have an orthogonality graph that can be decomposed into the disjoint union complete bipartite graphs, but the converse is also true: every graph that is built from complete bipartite graphs in this way is the orthogonality graph of some set of product states. To see this, on each party assign to each complete bipartite graph a distinct basis of $\bb{C}^2$ in the obvious way. For example, one set of product states giving rise to the orthogonality graph depicted in Figure~\ref{fig:2dim_eg} is as follows:
\begin{align*}
	\ket{v_0} & := \ket{0} \otimes \ket{0} \otimes \ket{0}, & \ket{v_1} & := \ket{0} \otimes \ket{1} \otimes \ket{+}, & \ket{v_2} & := \ket{0} \otimes \ket{0} \otimes \ket{1}, \\
	\ket{v_3} & := \ket{1} \otimes \ket{+} \otimes \ket{-}, & \ket{v_4} & := \ket{1} \otimes \ket{+} \otimes \ket{+}, & \ket{v_5} & := \ket{1} \otimes \ket{-} \otimes \ket{b}, \\
	\ket{v_6} & := \ket{1} \otimes \ket{-} \otimes \ket{b^\perp},
\end{align*}
where $\ket{+} := \frac{1}{\sqrt{2}}(\ket{0} + \ket{1})$, $\ket{-} := \frac{1}{\sqrt{2}}(\ket{0} - \ket{1})$, and $\{\ket{b},\ket{b^\perp}\}$ is any orthonormal basis of $\bb{C}^2$ not equal to $\{\ket{0},\ket{1}\}$ or $\{\ket{+},\ket{-}\}$.

It is often useful to draw orthogonality graphs of sets of qubit product states in a form that makes their decomposition in terms of complete bipartite graphs more transparent -- we draw shaded regions indicating which vertices are equal to each other (up to complex phase) on the given party, and lines between shaded regions indicate that all states in one of the regions are orthogonal to all states in the other region on that party (see Figure~\ref{fig:2dim_eg_compact}).
\begin{figure}[htb]
	\centering
	\begin{tikzpicture}[x=1.4cm, y=1.4cm, label distance=0.07cm]     
    \draw[draw=black] (-0.5840,0.4680) -- (0.4380,-0.3485);

    \draw[draw=black] (3.0125,0.3875) -- (2.723,0.629);
    \draw[draw=black] (3.5000,-0.8990) -- (4.3760,0.2020);

    \draw[draw=black,rounded corners=0.8cm] (7,1) -- (5.45,1.25) -- (6.025,-0.225);
    \draw[draw=black] (6.566,-0.899) -- (7,-0.5);
    \draw[draw=black] (7.975,-0.225) -- (7.777,0.629);

    \filldraw[fill=lightgray,line width=0.44cm,line join=round,draw=black] (0,1) -- (-0.777,0.629) -- (-0.975,-0.225) -- cycle;
    \filldraw[fill=lightgray,line width=0.4cm,line join=round,draw=lightgray] (0,1) -- (-0.777,0.629) -- (-0.975,-0.225) -- cycle;
    \filldraw[fill=lightgray,line width=0.44cm,line join=round,draw=black] (-0.434,-0.899) -- (0.434,-0.899) -- (0.975,-0.225) -- (0.777,0.629) -- cycle;
    \filldraw[fill=lightgray,line width=0.4cm,line join=round,draw=lightgray] (-0.434,-0.899) -- (0.434,-0.899) -- (0.975,-0.225) -- (0.777,0.629) -- cycle;
    
    \draw[line width=0.02cm,draw=black,fill=lightgray] (2.723,0.629) circle (0.2cm);
    \draw[line width=0.02cm,draw=black,fill=lightgray] (7,1) circle (0.2cm);
	\draw[line width=0.02cm,draw=black,fill=lightgray] (6.025,-0.225) circle (0.2cm);
    \draw[line width=0.02cm,draw=black,fill=lightgray] (6.566,-0.899) circle (0.2cm);
	\draw[line width=0.02cm,draw=black,fill=lightgray] (7.975,-0.225) circle (0.2cm);
    \draw[line width=0.02cm,draw=black,fill=lightgray] (7.777,0.629) circle (0.2cm);
    \filldraw[fill=lightgray,line width=0.44cm,line join=round,draw=black] (3.5,1) -- (2.525,-0.225) -- (3.0125,0.3875) -- cycle;
    \filldraw[fill=lightgray,line width=0.4cm,line join=round,draw=lightgray] (3.5,1) -- (2.525,-0.225) -- (3.0125,0.3875) -- cycle;
    \filldraw[fill=lightgray,line width=0.44cm,line join=round,draw=black] (3.066,-0.899) -- (3.934,-0.899) -- (3.5000,-0.8990) -- cycle;
    \filldraw[fill=lightgray,line width=0.4cm,line join=round,draw=lightgray] (3.066,-0.899) -- (3.934,-0.899) -- (3.5000,-0.8990) -- cycle;
    \filldraw[fill=lightgray,line width=0.44cm,line join=round,draw=black] (4.475,-0.225) -- (4.277,0.629) -- (4.3760,0.2020) -- cycle;
    \filldraw[fill=lightgray,line width=0.4cm,line join=round,draw=lightgray] (4.475,-0.225) -- (4.277,0.629) -- (4.3760,0.2020) -- cycle;
	\filldraw[fill=lightgray,line width=0.44cm,line join=round,draw=black] (6.223,0.629) -- (7.434,-0.899) -- (6.8285,-0.1350)-- cycle;
    \filldraw[fill=lightgray,line width=0.4cm,line join=round,draw=lightgray] (6.223,0.629) -- (7.434,-0.899) -- (6.8285,-0.1350) -- cycle;
        
    \vertex[fill] (v00) at (0,1) [label=90:$v_{0}$]{};
		\vertex[fill] (v01) at (-0.777,0.629) [label=141:$v_{1}$]{};
		\vertex[fill] (v02) at (-0.975,-0.225) [label=193:$v_{2}$]{};
		\vertex[fill] (v03) at (-0.434,-0.899) [label=244:$v_{3}$]{};
		\vertex[fill] (v04) at (0.434,-0.899) [label=296:$v_{4}$]{};
		\vertex[fill] (v05) at (0.975,-0.225) [label=347:$v_{5}$]{};
		\vertex[fill] (v06) at (0.777,0.629) [label=39:$v_{6}$]{};
				
		\vertex[fill] (v10) at (3.5,1) [label=90:$v_{0}$]{};
		\vertex[fill] (v11) at (2.723,0.629) [label=141:$v_{1}$]{};
		\vertex[fill] (v12) at (2.525,-0.225) [label=193:$v_{2}$]{};
		\vertex[fill] (v13) at (3.066,-0.899) [label=244:$v_{3}$]{};
		\vertex[fill] (v14) at (3.934,-0.899) [label=296:$v_{4}$]{};
		\vertex[fill] (v15) at (4.475,-0.225) [label=347:$v_{5}$]{};
		\vertex[fill] (v16) at (4.277,0.629) [label=39:$v_{6}$]{};

		\vertex[fill] (v20) at (7,1) [label=90:$v_{0}$]{};
		\vertex[fill] (v21) at (6.223,0.629) [label=141:$v_{1}$]{};
		\vertex[fill] (v22) at (6.025,-0.225) [label=193:$v_{2}$]{};
		\vertex[fill] (v23) at (6.566,-0.899) [label=244:$v_{3}$]{};
		\vertex[fill] (v24) at (7.434,-0.899) [label=296:$v_{4}$]{};
		\vertex[fill] (v25) at (7.975,-0.225) [label=347:$v_{5}$]{};
		\vertex[fill] (v26) at (7.777,0.629) [label=39:$v_{6}$]{};
	\end{tikzpicture}
	\caption{A representation of the same orthogonality graph as that of Figure~\ref{fig:2dim_eg}. Vertices within the same shaded region represent states that are equal to each other on that party. Lines between shaded regions indicate that every state within one of the regions is orthogonal to every state within the other region.}\label{fig:2dim_eg_compact}
\end{figure}
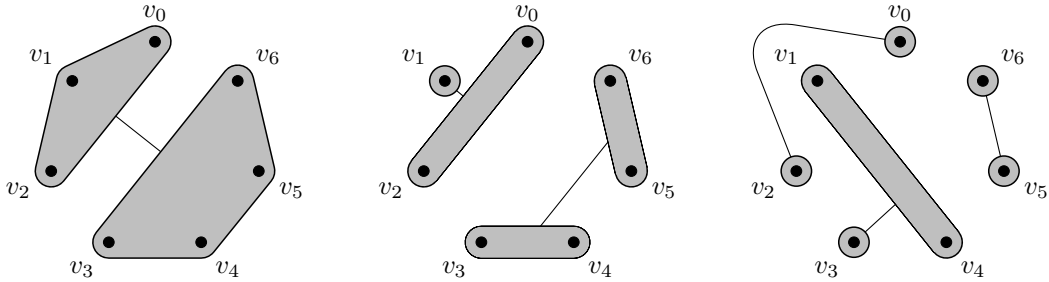

It now becomes straightforward to see whether or not a product basis is unextendible just by looking at its orthogonality graph. A set of product states is unextendible if and only if there is no way to choose one shaded region on each party such that every vertex $v_0,v_1,\ldots,v_{s-1}$ is contained within at least one of the shaded regions. For example, the set of product states described by Figure~\ref{fig:2dim_eg_compact} is extendible because we can choose the shaded region containing $v_3, v_4, v_5, v_6$ on the first subsystem, $v_0, v_2$ on the second subsystem, and $v_1,v_4$ on the third subsystem.

The following simple lemma shows that, in an orthogonality graph of a UPB, every shaded region must be connected to exactly one other shaded region via an edge.
\begin{lemma}\label{lem:all_pos}
	If $\cl{S} \subseteq (\bb{C}^2)^{\otimes p}$ is a UPB, then for all $\ket{v} \in \cl{S}$ and all integers $1 \leq j \leq p$ there is another product state $\ket{w} \in \cl{S}$ such that $\ket{v}$ and $\ket{w}$ are orthogonal on the $j$-th subsystem.
\end{lemma}
\begin{proof}
	Suppose that there exists $1 \leq j \leq p$ and $\ket{v} := \ket{v_{(1)}} \otimes \cdots \otimes \ket{v_{(p)}} \in \cl{S}$ such that $\ket{v}$ is not orthogonal to any other member of $\cl{S}$ on the $j$-th subsystem. Because $\cl{S}$ is a product basis, $\ket{v}$ must be orthogonal to every member of $\cl{S}$ on the remaining $p-1$ subsystems. It follows that if $\ket{v_{(j)}^\perp}$ is orthogonal to $\ket{v_{(j)}}$ then the product state $\ket{v_{(1)}} \otimes \cdots \ket{v_{(j-1)}} \otimes \ket{v_{(j)}^\perp} \otimes \ket{v_{(j+1)}} \otimes \cdots \otimes \ket{v_{(p)}}$ is orthogonal to every element of $\cl{S}$, which shows that $\cl{S}$ is extendible.
\end{proof}

An obvious corollary of Lemma~\ref{lem:all_pos} is that, in the orthogonality graph of a UPB, every party must have an even number of distinct shaded regions -- a fact that will be very useful in Section~\ref{sec:minimality_proof}.

\section{Construction of Small UPBs}\label{sec:construction}

Recall that our goal is to show that the smallest UPB in $(\bb{C}^2)^{\otimes 8}$ consists of $11$ states and the smallest UPB in $(\bb{C}^2)^{\otimes 4k}$ consists of $4k + 4$ states when $k \geq 3$. Our first step toward this goal is to construct a UPB of the desired size in these cases.
\begin{lemma}\label{lem:8_min_construct}
	There exists a UPB in $(\bb{C}^2)^{\otimes 8}$ consisting of $11$ states. 
\end{lemma}
\begin{proof}
	The result follows simply from demonstrating an orthogonality graph on $11$ vertices that satisfies the product basis and unextendibility requirements described in Section~\ref{sec:prelims}. Such an orthogonality graph is provided in Figure~\ref{fig:upb_8_11}.
	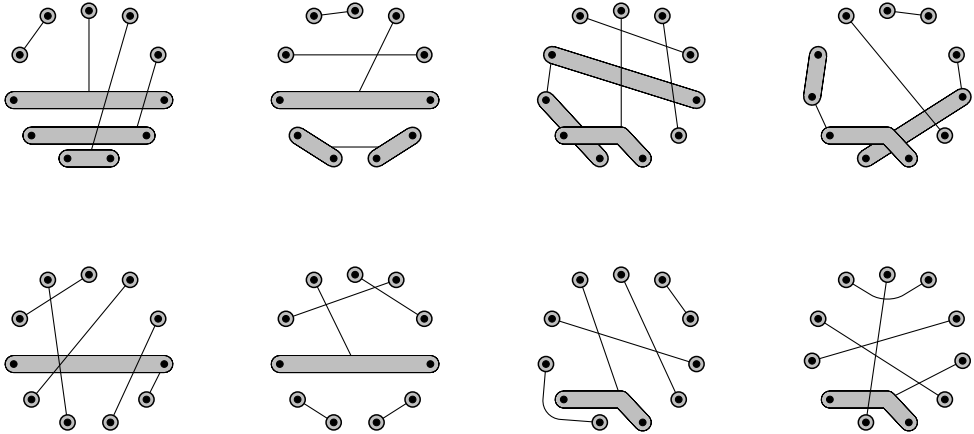
\begin{figure}[htb]
	\centering
	\begin{tikzpicture}[x=1cm, y=1cm, label distance=0cm]  
    \draw[draw=black] (2.9810,-14.8070) -- (4.0190,-14.8070);
    \draw[draw=black] (6.090,-13.585) -- (6.010,-14.185);
    \draw[draw=black] (9.510,-14.142) -- (9.744,-14.655);
    \draw[draw=black] (0,-13) -- (0,-14.185);
	\draw[draw=black] (-0.541,-13.069) -- (-0.910,-13.585);

    \filldraw[fill=lightgray,line width=0.24cm,line join=round,draw=black] (-0.990,-14.185) -- (0.990,-14.185) -- (0,-14.185) -- cycle;
    \filldraw[fill=lightgray,line width=0.20cm,line join=round,draw=lightgray] (-0.990,-14.185) -- (0.990,-14.185) -- (0,-14.185) -- cycle;

	\draw[draw=black] (0.910,-13.585) -- (0.6,-14.655);
    \filldraw[fill=lightgray,line width=0.24cm,line join=round,draw=black] (-0.756,-14.655) -- (0.756,-14.655) -- (0,-14.655) -- cycle;
    \filldraw[fill=lightgray,line width=0.20cm,line join=round,draw=lightgray] (-0.756,-14.655) -- (0.756,-14.655) -- (0,-14.655) -- cycle;

	\draw[draw=black] (0.541,-13.069) -- (0,-14.959);
	\filldraw[fill=lightgray,line width=0.24cm,line join=round,draw=black] (-0.282,-14.959) -- (0.282,-14.959) -- (0,-14.959) -- cycle;
    \filldraw[fill=lightgray,line width=0.20cm,line join=round,draw=lightgray] (-0.282,-14.959) -- (0.282,-14.959) -- (0,-14.959) -- cycle;

    \draw[line width=0.02cm,draw=black,fill=lightgray] (-0.541,-13.069) circle (0.1cm);
    \draw[line width=0.02cm,draw=black,fill=lightgray] (0,-13) circle (0.1cm);
    \draw[line width=0.02cm,draw=black,fill=lightgray] (-0.910,-13.585) circle (0.1cm);
    \draw[line width=0.02cm,draw=black,fill=lightgray] (0.910,-13.585) circle (0.1cm);
    \draw[line width=0.02cm,draw=black,fill=lightgray] (0.541,-13.069) circle (0.1cm);

		\tvertex[fill] (v000) at (0,-13) []{};
		\tvertex[fill] (v001) at (-0.541,-13.069) []{};
		\tvertex[fill] (v002) at (-0.910,-13.585) []{};
		\tvertex[fill] (v003) at (-0.990,-14.185) []{};
		\tvertex[fill] (v004) at (-0.756,-14.655) []{};
		\tvertex[fill] (v005) at (-0.282,-14.959) []{};
		\tvertex[fill] (v006) at (0.282,-14.959) []{};
		\tvertex[fill] (v007) at (0.756,-14.655) []{};
		\tvertex[fill] (v008) at (0.990,-14.185) []{};
		\tvertex[fill] (v009) at (0.910,-13.585) []{};
		\tvertex[fill] (v010) at (0.541,-13.069) []{};

	\draw[draw=black] (3.5,-13) -- (2.959,-13.069);
	\draw[draw=black] (2.590,-13.585) -- (4.410,-13.585);
	\draw[draw=black] (4.041,-13.069) -- (3.5,-14.185);

    \filldraw[fill=lightgray,line width=0.24cm,line join=round,draw=black] (2.51,-14.185) -- (4.490,-14.185) -- (3.5000,-14.1850) -- cycle;
    \filldraw[fill=lightgray,line width=0.20cm,line join=round,draw=lightgray] (2.51,-14.185) -- (4.490,-14.185) -- (3.5000,-14.1850) -- cycle;
    \filldraw[fill=lightgray,line width=0.24cm,line join=round,draw=black] (2.744,-14.655) -- (3.218,-14.959) -- (2.9810,-14.8070) -- cycle;
    \filldraw[fill=lightgray,line width=0.20cm,line join=round,draw=lightgray] (2.744,-14.655) -- (3.218,-14.959) -- (2.9810,-14.8070) -- cycle;
	\filldraw[fill=lightgray,line width=0.24cm,line join=round,draw=black] (3.782,-14.959) -- (4.256,-14.655) -- (4.0190,-14.8070) -- cycle;
    \filldraw[fill=lightgray,line width=0.20cm,line join=round,draw=lightgray] (3.782,-14.959) -- (4.256,-14.655) -- (4.0190,-14.8070) -- cycle;

    \draw[line width=0.02cm,draw=black,fill=lightgray] (2.959,-13.069) circle (0.1cm);
    \draw[line width=0.02cm,draw=black,fill=lightgray] (3.5,-13) circle (0.1cm);
    \draw[line width=0.02cm,draw=black,fill=lightgray] (2.59,-13.585) circle (0.1cm);
    \draw[line width=0.02cm,draw=black,fill=lightgray] (4.410,-13.585) circle (0.1cm);
    \draw[line width=0.02cm,draw=black,fill=lightgray] (4.041,-13.069) circle (0.1cm);

		\tvertex[fill] (v100) at (3.5,-13) []{};
		\tvertex[fill] (v101) at (2.959,-13.069) []{};
		\tvertex[fill] (v102) at (2.590,-13.585) []{};
		\tvertex[fill] (v103) at (2.510,-14.185) []{};
		\tvertex[fill] (v104) at (2.744,-14.655) []{};
		\tvertex[fill] (v105) at (3.218,-14.959) []{};
		\tvertex[fill] (v106) at (3.782,-14.959) []{};
		\tvertex[fill] (v107) at (4.256,-14.655) []{};
		\tvertex[fill] (v108) at (4.490,-14.185) []{};
		\tvertex[fill] (v109) at (4.410,-13.585) []{};
		\tvertex[fill] (v110) at (4.041,-13.069) []{};

    \filldraw[fill=lightgray,line width=0.24cm,line join=round,draw=black] (6.010,-14.185) -- (6.718,-14.959) -- (6.3640,-14.5720) -- cycle;
    \filldraw[fill=lightgray,line width=0.20cm,line join=round,draw=lightgray] (6.010,-14.185) -- (6.718,-14.959) -- (6.3640,-14.5720) -- cycle;
	\filldraw[fill=lightgray,line width=0.24cm,line join=round,draw=black] (6.090,-13.585) -- (7.990,-14.185) -- (7.0400,-13.8850) -- cycle;
    \filldraw[fill=lightgray,line width=0.20cm,line join=round,draw=lightgray] (6.090,-13.585) -- (7.990,-14.185) -- (7.0400,-13.8850) -- cycle;
	
	\draw[draw=black] (7,-13) -- (7,-14.655);
	\draw[draw=black] (7.541,-13.069) -- (7.756,-14.655);
	\draw[draw=black] (6.459,-13.069) -- (7.910,-13.585);
	\filldraw[fill=lightgray,line width=0.24cm,line join=round,draw=black] (7,-14.655) -- (6.244,-14.655) -- (7,-14.655) -- (7.282,-14.959) -- cycle;
    \filldraw[fill=lightgray,line width=0.20cm,line join=round,draw=lightgray] (7,-14.655) -- (6.244,-14.655) -- (7,-14.655) -- (7.282,-14.959) -- cycle;

    \draw[line width=0.02cm,draw=black,fill=lightgray] (6.459,-13.069) circle (0.1cm);
    \draw[line width=0.02cm,draw=black,fill=lightgray] (7,-13) circle (0.1cm);
    \draw[line width=0.02cm,draw=black,fill=lightgray] (7.756,-14.655) circle (0.1cm);
    \draw[line width=0.02cm,draw=black,fill=lightgray] (7.910,-13.585) circle (0.1cm);
    \draw[line width=0.02cm,draw=black,fill=lightgray] (7.541,-13.069) circle (0.1cm);

		\tvertex[fill] (v200) at (7,-13) []{};
		\tvertex[fill] (v201) at (6.459,-13.069) []{};
		\tvertex[fill] (v202) at (6.090,-13.585) []{};
		\tvertex[fill] (v203) at (6.010,-14.185) []{};
		\tvertex[fill] (v204) at (6.244,-14.655) []{};
		\tvertex[fill] (v205) at (6.718,-14.959) []{};
		\tvertex[fill] (v206) at (7.282,-14.959) []{};
		\tvertex[fill] (v207) at (7.756,-14.655) []{};
		\tvertex[fill] (v208) at (7.990,-14.185) []{};
		\tvertex[fill] (v209) at (7.910,-13.585) []{};
		\tvertex[fill] (v210) at (7.541,-13.069) []{};

	\draw[draw=black] (11.490,-14.142) -- (11.410,-13.585);
	\draw[draw=black] (10.5,-13) -- (11.041,-13.069);

    \filldraw[fill=lightgray,line width=0.24cm,line join=round,draw=black] (9.590,-13.585) -- (9.510,-14.142) -- (9.5500,-13.8635) -- cycle;
    \filldraw[fill=lightgray,line width=0.20cm,line join=round,draw=lightgray] (9.590,-13.585) -- (9.510,-14.142) -- (9.5500,-13.8635) -- cycle;
    \filldraw[fill=lightgray,line width=0.24cm,line join=round,draw=black] (10.218,-14.959) -- (11.490,-14.142) -- (10.8540,-14.5505) -- cycle;
    \filldraw[fill=lightgray,line width=0.20cm,line join=round,draw=lightgray] (10.218,-14.959) -- (11.490,-14.142) -- (10.8540,-14.5505) -- cycle;
	\filldraw[fill=lightgray,line width=0.24cm,line join=round,draw=black] (10.5,-14.655) -- (9.744,-14.655) -- (10.5,-14.655) -- (10.782,-14.959) -- cycle;
    \filldraw[fill=lightgray,line width=0.20cm,line join=round,draw=lightgray] (10.5,-14.655) -- (9.744,-14.655) -- (10.5,-14.655) -- (10.782,-14.959) -- cycle;
    
    \draw[draw=black] (9.959,-13.069) -- (11.256,-14.655);

    \draw[line width=0.02cm,draw=black,fill=lightgray] (9.959,-13.069) circle (0.1cm);
    \draw[line width=0.02cm,draw=black,fill=lightgray] (10.5,-13) circle (0.1cm);
    \draw[line width=0.02cm,draw=black,fill=lightgray] (11.256,-14.655) circle (0.1cm);
    \draw[line width=0.02cm,draw=black,fill=lightgray] (11.410,-13.585) circle (0.1cm);
    \draw[line width=0.02cm,draw=black,fill=lightgray] (11.041,-13.069) circle (0.1cm);

		\tvertex[fill] (v300) at (10.5,-13) []{};
		\tvertex[fill] (v301) at (9.959,-13.069) []{};
		\tvertex[fill] (v302) at (9.590,-13.585) []{};
		\tvertex[fill] (v303) at (9.510,-14.142) []{};
		\tvertex[fill] (v304) at (9.744,-14.655) []{};
		\tvertex[fill] (v305) at (10.218,-14.959) []{};
		\tvertex[fill] (v306) at (10.782,-14.959) []{};
		\tvertex[fill] (v307) at (11.256,-14.655) []{};
		\tvertex[fill] (v308) at (11.490,-14.142) []{};
		\tvertex[fill] (v309) at (11.410,-13.585) []{};
		\tvertex[fill] (v310) at (11.041,-13.069) []{};


	\draw[draw=black] (0.756,-18.155) -- (0.990,-17.685);
	\draw[draw=black] (0,-16.5) -- (-0.910,-17.085);

    \filldraw[fill=lightgray,line width=0.24cm,line join=round,draw=black] (-0.990,-17.685) -- (0.990,-17.685) -- (0,-17.6850) -- cycle;
    \filldraw[fill=lightgray,line width=0.20cm,line join=round,draw=lightgray] (-0.990,-17.685) -- (0.990,-17.685) -- (0,-17.6850) -- cycle;
    
	\draw[draw=black] (-0.756,-18.155) -- (0.541,-16.569);
	\draw[draw=black] (-0.541,-16.569) -- (-0.282,-18.459);
	\draw[draw=black] (0.282,-18.459) -- (0.910,-17.085);

    \draw[line width=0.02cm,draw=black,fill=lightgray] (-0.756,-18.155) circle (0.1cm);
    \draw[line width=0.02cm,draw=black,fill=lightgray] (-0.282,-18.459) circle (0.1cm);
    \draw[line width=0.02cm,draw=black,fill=lightgray] (-0.541,-16.569) circle (0.1cm);
    \draw[line width=0.02cm,draw=black,fill=lightgray] (0.282,-18.459) circle (0.1cm);
    \draw[line width=0.02cm,draw=black,fill=lightgray] (0.756,-18.155) circle (0.1cm);
    \draw[line width=0.02cm,draw=black,fill=lightgray] (0,-16.5) circle (0.1cm);
    \draw[line width=0.02cm,draw=black,fill=lightgray] (-0.910,-17.085) circle (0.1cm);
    \draw[line width=0.02cm,draw=black,fill=lightgray] (0.910,-17.085) circle (0.1cm);
    \draw[line width=0.02cm,draw=black,fill=lightgray] (0.541,-16.569) circle (0.1cm);

		\tvertex[fill] (v000) at (0,-16.5) []{};
		\tvertex[fill] (v001) at (-0.541,-16.569) []{};
		\tvertex[fill] (v002) at (-0.910,-17.085) []{};
		\tvertex[fill] (v003) at (-0.990,-17.685) []{};
		\tvertex[fill] (v004) at (-0.756,-18.155) []{};
		\tvertex[fill] (v005) at (-0.282,-18.459) []{};
		\tvertex[fill] (v006) at (0.282,-18.459) []{};
		\tvertex[fill] (v007) at (0.756,-18.155) []{};
		\tvertex[fill] (v008) at (0.990,-17.685) []{};
		\tvertex[fill] (v009) at (0.910,-17.085) []{};
		\tvertex[fill] (v010) at (0.541,-16.569) []{};

	\draw[draw=black] (2.744,-18.155) -- (3.218,-18.459);
	\draw[draw=black] (3.782,-18.459) -- (4.256,-18.155);
	\draw[draw=black] (2.959,-16.569) -- (3.5,-17.685);
	\draw[draw=black] (2.590,-17.085) -- (4.041,-16.569);
	\draw[draw=black] (3.5,-16.5) -- (4.410,-17.085);

    \filldraw[fill=lightgray,line width=0.24cm,line join=round,draw=black] (2.51,-17.685) -- (4.490,-17.685) -- (3.5,-17.685) -- cycle;
    \filldraw[fill=lightgray,line width=0.20cm,line join=round,draw=lightgray] (2.51,-17.685) -- (4.490,-17.685) -- (3.5,-17.685) -- cycle;

    \draw[line width=0.02cm,draw=black,fill=lightgray] (2.7441,-18.155) circle (0.1cm);
    \draw[line width=0.02cm,draw=black,fill=lightgray] (3.218,-18.459) circle (0.1cm);
    \draw[line width=0.02cm,draw=black,fill=lightgray] (4.256,-18.155) circle (0.1cm);
    \draw[line width=0.02cm,draw=black,fill=lightgray] (2.959,-16.569) circle (0.1cm);
    \draw[line width=0.02cm,draw=black,fill=lightgray] (3.5,-16.5) circle (0.1cm);
    \draw[line width=0.02cm,draw=black,fill=lightgray] (2.59,-17.085) circle (0.1cm);
    \draw[line width=0.02cm,draw=black,fill=lightgray] (4.410,-17.085) circle (0.1cm);
    \draw[line width=0.02cm,draw=black,fill=lightgray] (4.041,-16.569) circle (0.1cm);
    \draw[line width=0.02cm,draw=black,fill=lightgray] (3.782,-18.459) circle (0.1cm);

		\tvertex[fill] (v100) at (3.5,-16.5) []{};
		\tvertex[fill] (v101) at (2.959,-16.569) []{};
		\tvertex[fill] (v102) at (2.590,-17.085) []{};
		\tvertex[fill] (v103) at (2.510,-17.685) []{};
		\tvertex[fill] (v104) at (2.744,-18.155) []{};
		\tvertex[fill] (v105) at (3.218,-18.459) []{};
		\tvertex[fill] (v106) at (3.782,-18.459) []{};
		\tvertex[fill] (v107) at (4.256,-18.155) []{};
		\tvertex[fill] (v108) at (4.490,-17.685) []{};
		\tvertex[fill] (v109) at (4.410,-17.085) []{};
		\tvertex[fill] (v110) at (4.041,-16.569) []{};

	\draw[draw=black,rounded corners=0.25cm] (6.010,-17.685) -- (5.95,-18.4) -- (6.718,-18.459);
	\draw[draw=black] (6.459,-16.569) -- (7,-18.155);
	\draw[draw=black] (6.090,-17.085) -- (7.990,-17.685);
	\draw[draw=black] (7,-16.5) -- (7.756,-18.155);
	\draw[draw=black] (7.541,-16.569) -- (7.910,-17.085);

    \filldraw[fill=lightgray,line width=0.24cm,line join=round,draw=black] (7,-18.155) -- (6.244,-18.155) -- (7,-18.155) -- (7.282,-18.459) -- cycle;
    \filldraw[fill=lightgray,line width=0.20cm,line join=round,draw=lightgray] (7,-18.155) -- (6.244,-18.155) -- (7,-18.155) -- (7.282,-18.459) -- cycle;

    \draw[line width=0.02cm,draw=black,fill=lightgray] (6.090,-17.085) circle (0.1cm);
    \draw[line width=0.02cm,draw=black,fill=lightgray] (6.010,-17.685) circle (0.1cm);
    \draw[line width=0.02cm,draw=black,fill=lightgray] (7.910,-17.085) circle (0.1cm);
    \draw[line width=0.02cm,draw=black,fill=lightgray] (7.990,-17.685) circle (0.1cm);
    \draw[line width=0.02cm,draw=black,fill=lightgray] (6.459,-16.569) circle (0.1cm);
    \draw[line width=0.02cm,draw=black,fill=lightgray] (7,-16.5) circle (0.1cm);
    \draw[line width=0.02cm,draw=black,fill=lightgray] (7.756,-18.155) circle (0.1cm);
    \draw[line width=0.02cm,draw=black,fill=lightgray] (6.718,-18.459) circle (0.1cm);
    \draw[line width=0.02cm,draw=black,fill=lightgray] (7.541,-16.569) circle (0.1cm);

		\tvertex[fill] (v200) at (7,-16.5) []{};
		\tvertex[fill] (v201) at (6.459,-16.569) []{};
		\tvertex[fill] (v202) at (6.090,-17.085) []{};
		\tvertex[fill] (v203) at (6.010,-17.685) []{};
		\tvertex[fill] (v204) at (6.244,-18.155) []{};
		\tvertex[fill] (v205) at (6.718,-18.459) []{};
		\tvertex[fill] (v206) at (7.282,-18.459) []{};
		\tvertex[fill] (v207) at (7.756,-18.155) []{};
		\tvertex[fill] (v208) at (7.990,-17.685) []{};
		\tvertex[fill] (v209) at (7.910,-17.085) []{};
		\tvertex[fill] (v210) at (7.541,-16.569) []{};

	\draw[draw=black,rounded corners=0.25cm] (9.959,-16.569) -- (10.5,-16.9) -- (11.041,-16.569);
	\draw[draw=black] (9.590,-17.085) -- (11.256,-18.155);
	\draw[draw=black] (9.510,-17.642) -- (11.410,-17.085);
	\draw[draw=black] (10.5,-18.155) -- (11.490,-17.642);

	\filldraw[fill=lightgray,line width=0.24cm,line join=round,draw=black] (10.5,-18.155) -- (9.744,-18.155) -- (10.5,-18.155) -- (10.782,-18.459) -- cycle;
    \filldraw[fill=lightgray,line width=0.20cm,line join=round,draw=lightgray] (10.5,-18.155) -- (9.744,-18.155) -- (10.5,-18.155) -- (10.782,-18.459) -- cycle;

	\draw[draw=black] (10.5,-16.5) -- (10.218,-18.459);

    \draw[line width=0.02cm,draw=black,fill=lightgray] (10.218,-18.459) circle (0.1cm);
    \draw[line width=0.02cm,draw=black,fill=lightgray] (9.590,-17.085) circle (0.1cm);
    \draw[line width=0.02cm,draw=black,fill=lightgray] (9.510,-17.642) circle (0.1cm);
    \draw[line width=0.02cm,draw=black,fill=lightgray] (11.410,-17.085) circle (0.1cm);
    \draw[line width=0.02cm,draw=black,fill=lightgray] (11.490,-17.642) circle (0.1cm);
    \draw[line width=0.02cm,draw=black,fill=lightgray] (9.959,-16.569) circle (0.1cm);
    \draw[line width=0.02cm,draw=black,fill=lightgray] (10.5,-16.5) circle (0.1cm);
    \draw[line width=0.02cm,draw=black,fill=lightgray] (11.256,-18.155) circle (0.1cm);
    \draw[line width=0.02cm,draw=black,fill=lightgray] (11.041,-16.569) circle (0.1cm);

		\tvertex[fill] (v300) at (10.5,-16.5) []{};
		\tvertex[fill] (v301) at (9.959,-16.569) []{};
		\tvertex[fill] (v302) at (9.590,-17.085) []{};
		\tvertex[fill] (v303) at (9.510,-17.642) []{};
		\tvertex[fill] (v304) at (9.744,-18.155) []{};
		\tvertex[fill] (v305) at (10.218,-18.459) []{};
		\tvertex[fill] (v306) at (10.782,-18.459) []{};
		\tvertex[fill] (v307) at (11.256,-18.155) []{};
		\tvertex[fill] (v308) at (11.490,-17.642) []{};
		\tvertex[fill] (v309) at (11.410,-17.085) []{};
		\tvertex[fill] (v310) at (11.041,-16.569) []{};
	\end{tikzpicture}
	\caption{Orthogonality graphs demonstrating that there exists an $11$-state UPB in $(\bb{C}^2)^{\otimes 8}$.}\label{fig:upb_8_11}
\end{figure}

Indeed, it is straightforward (albeit tedious) to check that the $8$ graphs depicted in Figure~\ref{fig:upb_8_11} contain all $55$ possible edges between $11$ vertices, so the corresponding product states are mutually orthogonal. Unextendibility follows from the (also straightforward but tedious) fact that there is no way to choose a shaded region containing $2$ vertices on $3$ different parties without at least $2$ of them containing the same vertex.
\end{proof}

We note that the UPB of Lemma~\ref{lem:8_min_construct} was found by a combination of computer search and tweaking by hand, and it does not seem to generalize to other values of $p$ in any natural way. On the other hand, the UPBs that we now construct of cardinality $4k+4$ are much ``tidier''.

\begin{lemma}\label{lem:4k_min_construct}
	If $k \geq 2$ then there exists a UPB in $(\bb{C}^2)^{\otimes 4k}$ consisting of $4k + 4$ states. 
\end{lemma}
\begin{proof}
	We begin by defining a family of $k+1$ graphs $B_{j,k} := (V,E_j)$ for $0 \leq j \leq k$, each on the same set of $4k+4$ vertices $V := \{v_i,w_i,x_i,y_i, : 0 \leq i \leq k \}$. The set of edges $E_j$ in the graph $B_{j,k}$ is defined as follows:
	\begin{align*}
		E_j & := \big\{ (v_i,x_{(i+j)(\text{mod }(k+1))}), (v_i,y_{(i+j)(\text{mod }(k+1))}), \\
		& \quad \quad \quad (w_i,x_{(i+j)(\text{mod }(k+1))}), (w_i,y_{(i+j)(\text{mod }(k+1))}) : 0 \leq i \leq k \big\}.
	\end{align*}

The three graphs $B_{0,2}$, $B_{1,2}$, and $B_{2,2}$ in the $k = 2$ case are depicted in Figure~\ref{fig:B_graphs}. It is clear that the graph obtained by taking the union of all edges in all sets $B_{j,k}$ for $0 \leq j \leq k$ is $K_{2k+2,2k+2}$, the complete bipartite graph on two sets of $2k+2$ vertices.
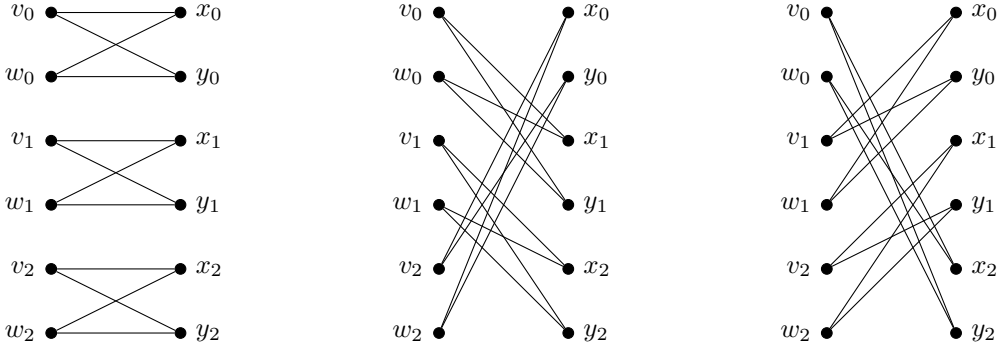
\begin{figure}[htb]
	\centering
	\begin{tikzpicture}[x=1.7cm, y=0.85cm, label distance=0cm]
		\vertex[fill] (y0) at (0,4) [label=0:$y_{0}$]{};
		\vertex[fill] (x0) at (0,5) [label=0:$x_{0}$]{};
		\vertex[fill] (w0) at (-1,4) [label=180:$w_{0}$]{};
		\vertex[fill] (v0) at (-1,5) [label=180:$v_{0}$]{};
		
		\vertex[fill] (y1) at (0,2) [label=0:$y_{1}$]{};
		\vertex[fill] (x1) at (0,3) [label=0:$x_{1}$]{};
		\vertex[fill] (w1) at (-1,2) [label=180:$w_{1}$]{};
		\vertex[fill] (v1) at (-1,3) [label=180:$v_{1}$]{};
		
		\vertex[fill] (y2) at (0,0) [label=0:$y_{2}$]{};
		\vertex[fill] (x2) at (0,1) [label=0:$x_{2}$]{};
		\vertex[fill] (w2) at (-1,0) [label=180:$w_{2}$]{};
		\vertex[fill] (v2) at (-1,1) [label=180:$v_{2}$]{};
		
		\vertex[fill] (y0b) at (3,4) [label=0:$y_{0}$]{};
		\vertex[fill] (x0b) at (3,5) [label=0:$x_{0}$]{};
		\vertex[fill] (w0b) at (2,4) [label=180:$w_{0}$]{};
		\vertex[fill] (v0b) at (2,5) [label=180:$v_{0}$]{};
		
		\vertex[fill] (y1b) at (3,2) [label=0:$y_{1}$]{};
		\vertex[fill] (x1b) at (3,3) [label=0:$x_{1}$]{};
		\vertex[fill] (w1b) at (2,2) [label=180:$w_{1}$]{};
		\vertex[fill] (v1b) at (2,3) [label=180:$v_{1}$]{};
		
		\vertex[fill] (y2b) at (3,0) [label=0:$y_{2}$]{};
		\vertex[fill] (x2b) at (3,1) [label=0:$x_{2}$]{};
		\vertex[fill] (w2b) at (2,0) [label=180:$w_{2}$]{};
		\vertex[fill] (v2b) at (2,1) [label=180:$v_{2}$]{};
	
		\vertex[fill] (y0c) at (6,4) [label=0:$y_{0}$]{};
		\vertex[fill] (x0c) at (6,5) [label=0:$x_{0}$]{};
		\vertex[fill] (w0c) at (5,4) [label=180:$w_{0}$]{};
		\vertex[fill] (v0c) at (5,5) [label=180:$v_{0}$]{};
		
		\vertex[fill] (y1c) at (6,2) [label=0:$y_{1}$]{};
		\vertex[fill] (x1c) at (6,3) [label=0:$x_{1}$]{};
		\vertex[fill] (w1c) at (5,2) [label=180:$w_{1}$]{};
		\vertex[fill] (v1c) at (5,3) [label=180:$v_{1}$]{};
		
		\vertex[fill] (y2c) at (6,0) [label=0:$y_{2}$]{};
		\vertex[fill] (x2c) at (6,1) [label=0:$x_{2}$]{};
		\vertex[fill] (w2c) at (5,0) [label=180:$w_{2}$]{};
		\vertex[fill] (v2c) at (5,1) [label=180:$v_{2}$]{};
		
		\path 
			(v0) edge (x0)
			(v0) edge (y0)
			(w0) edge (x0)
			(w0) edge (y0)
			
			(v1) edge (x1)
			(v1) edge (y1)
			(w1) edge (x1)
			(w1) edge (y1)

			(v2) edge (x2)
			(v2) edge (y2)
			(w2) edge (x2)
			(w2) edge (y2)

			(v0b) edge (x1b)
			(v0b) edge (y1b)
			(w0b) edge (x1b)
			(w0b) edge (y1b)
			
			(v1b) edge (x2b)
			(v1b) edge (y2b)
			(w1b) edge (x2b)
			(w1b) edge (y2b)

			(v2b) edge (x0b)
			(v2b) edge (y0b)
			(w2b) edge (x0b)
			(w2b) edge (y0b)

			(v0c) edge (x2c)
			(v0c) edge (y2c)
			(w0c) edge (x2c)
			(w0c) edge (y2c)
			
			(v1c) edge (x0c)
			(v1c) edge (y0c)
			(w1c) edge (x0c)
			(w1c) edge (y0c)

			(v2c) edge (x1c)
			(v2c) edge (y1c)
			(w2c) edge (x1c)
			(w2c) edge (y1c)
		;
	\end{tikzpicture}
	\caption{The graphs $B_{0,2}$ (left), $B_{1,2}$ (center), and $B_{2,2}$ (right), used in the construction of a UPB of size $12$ in $(\bb{C}^2)^{\otimes 8}$.}\label{fig:B_graphs}
\end{figure}

We now define three sets of states $S^{(j)} = \{\ket{v_i^{(j)}}, \ket{w_i^{(j)}}, \ket{x_i^{(j)}}, \ket{y_i^{(j)}} : 0 \leq i \leq k \} \subseteq \bb{C}^2$ that have orthogonality graphs $B_{j,k}$ for $0 \leq j \leq 2$ respectively. To this end, let $\{\ket{b_i},\ket{b_i^\perp}\}_{i=0}^{2k+1}$ be distinct orthonormal bases of $\bb{C}^2$ (i.e., $\braket{b_i}{b_i^\perp} = 0$ for all $i$, but $|\braket{b_i}{b_j}|,|\braket{b_i}{b_j^\perp}|,|\braket{b_i^\perp}{b_j^\perp}| \notin \{0,1\}$ whenever $i \neq j$). Then let
\begin{align*}
	\ket{v_i^{(j)}} := \ket{w_i^{(j)}} := \ket{b_i} \quad \text{ and } \quad \ket{x_i^{(j)}} := \ket{y_i^{(j)}} := \ket{b_{(i-j)(\text{mod }(k+1))}^\perp},
\end{align*}
for $0 \leq j \leq 2$, which clearly results in the desired orthogonality graphs. Furthermore, each set $S^{(j)}$ has the property that any state $\ket{z} \in \bb{C}^2$ can be orthogonal to at most two elements of $S^{(j)}$ -- a fact that we will use later when discussing unextendibility.

For each of the remaining $k - 2$ graphs $B_{j,k}$ ($3 \leq j \leq k$), we construct sets of product states $S^{(2j-3,2j-2)} = \{\ket{v_i^{(2j-3,2j-2)}}, \ket{w_i^{(2j-3,2j-2)}}, \ket{x_i^{(2j-3,2j-2)}}, \ket{y_i^{(2j-3,2j-2)}} : 0 \leq i \leq k \} \subseteq \bb{C}^2 \otimes \bb{C}^2$ that have orthogonality graphs $B_{j,k}$ for $3 \leq j \leq k$. To this end, define
\begin{align*}
	\ket{v_i^{(2j-3,2j-2)}} & := \ket{b_i} \otimes \ket{b_i} \\
	\ket{w_i^{(2j-3,2j-2)}} & := \ket{b_{i+(k+1)}} \otimes \ket{b_{i+(k+1)}} \\
	\ket{x_i^{(2j-3,2j-2)}} & := \ket{b_{(i-j)(\text{mod }(k+1))}^\perp} \otimes \ket{b_{(i-j)(\text{mod }(k+1))+(k+1)}^\perp} \\
	\ket{y_i^{(2j-3,2j-2)}} & := \ket{b_{(i-j)(\text{mod }(k+1))+(k+1)}^\perp} \otimes \ket{b_{(i-j)(\text{mod }(k+1))}^\perp},
\end{align*}
which results in the desired orthogonality graphs.

We now turn our attention to the complement graph of $K_{2k+2,2k+2}$, which is simply the disjoint union of two disjoint copies of $K_{2k+2}$, the complete graph on $2k+2$ vertices. We denote this graph by $K_{2k+2}^2$, and it is depicted in the $k = 2$ case in Figure~\ref{fig:2complete_graphs}. The graph $K_{2k+2}^2$ will be the orthogonality graph of the remaining $4k - (3 + 2(k-2)) = 2k + 1$ parties.
\begin{figure}[htb]
	\centering
	\begin{tikzpicture}[x=1.5cm, y=1.5cm, label distance=0cm]
		\vertex[fill] (v0) at (0,1) [label=90:$v_0$]{};
		\vertex[fill] (v1) at (0.866,0.5) [label=15:$w_0$]{};
		\vertex[fill] (v2) at (0.866,-0.5) [label=345:$v_1$]{};
		\vertex[fill] (v3) at (0,-1) [label=270:$w_1$]{};
		\vertex[fill] (v4) at (-0.866,-0.5) [label=195:$v_2$]{};
		\vertex[fill] (v5) at (-0.866,0.5) [label=165:$w_2$]{};
		
		\vertex[fill] (w0) at (3,1) [label=90:$x_0$]{};
		\vertex[fill] (w1) at (3.866,0.5) [label=15:$y_0$]{};
		\vertex[fill] (w2) at (3.866,-0.5) [label=345:$x_1$]{};
		\vertex[fill] (w3) at (3,-1) [label=270:$y_1$]{};
		\vertex[fill] (w4) at (2.134,-0.5) [label=195:$x_2$]{};
		\vertex[fill] (w5) at (2.134,0.5) [label=165:$y_2$]{};
			
		\path 
			(v0) edge (v1)
			(v0) edge (v2)
			(v0) edge (v3)
			(v0) edge (v4)
			(v0) edge (v5)
			(v1) edge (v2)
			(v1) edge (v3)
			(v1) edge (v4)
			(v1) edge (v5)
			(v2) edge (v3)
			(v2) edge (v4)
			(v2) edge (v5)
			(v3) edge (v4)
			(v3) edge (v5)
			(v4) edge (v5)
			
			(w0) edge (w1)
			(w0) edge (w2)
			(w0) edge (w3)
			(w0) edge (w4)
			(w0) edge (w5)
			(w1) edge (w2)
			(w1) edge (w3)
			(w1) edge (w4)
			(w1) edge (w5)
			(w2) edge (w3)
			(w2) edge (w4)
			(w2) edge (w5)
			(w3) edge (w4)
			(w3) edge (w5)
			(w4) edge (w5)
		;
	\end{tikzpicture}
	\caption{The graph $K_{6}^2$ that is the disjoint union of two copies of $K_{6}$.}\label{fig:2complete_graphs}
\end{figure}
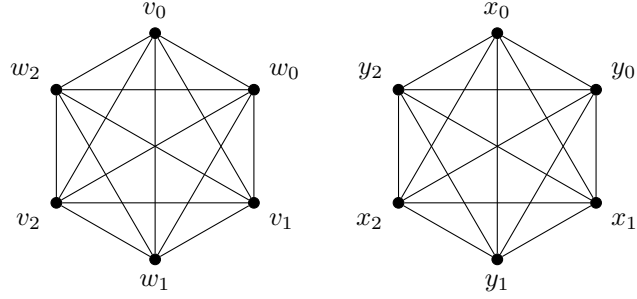

Our goal now is to define sets of states $S^{(j)} = \{\ket{v_i^{(j)}}, \ket{w_i^{(j)}}, \ket{x_i^{(j)}}, \ket{y_i^{(j)}} : 0 \leq i \leq k \} \subseteq \bb{C}^2$ for $2k-1 \leq j \leq 4k - 1$ such that their orthogonality graphs, when taken together, contain all edges of $K_{2k+2}^2$. To this end, we recall that it is well-known that $K_{2k+2}$ always has a $1$-factorization \cite[Theorem~9.1]{Har69}, so $K_{2k+2}^2$ clearly has a $1$-factorization as well (see Figure~\ref{fig:2complete_graphs_1fac}). This $1$-factorization decomposes $K_{2k+2}^2$ into $2k+1$ distinct $1$-regular spanning subgraphs, and any such graph is clearly the orthogonality graph of the set of states $\{\ket{b_0},\ket{b_0^\perp},\ldots,\ket{b_{2k+1}},\ket{b_{2k+1}^\perp}\} \subset \bb{C}^2$ (under an appropriate labelling of the vertices).

\begin{figure}[htb]
	\centering
	\begin{tikzpicture}[x=0.9cm, y=0.9cm, label distance=0cm]
    \draw[draw=black] (5.5,0.5) -- (5.5,-8.5);

		\vertex[fill] (v00) at (0,1) [label=90:$v_0$]{};
		\vertex[fill] (v01) at (0.866,0.5) [label=15:$w_0$]{};
		\vertex[fill] (v02) at (0.866,-0.5) [label=345:$v_1$]{};
		\vertex[fill] (v03) at (0,-1) [label=270:$w_1$]{};
		\vertex[fill] (v04) at (-0.866,-0.5) [label=195:$v_2$]{};
		\vertex[fill] (v05) at (-0.866,0.5) [label=165:$w_2$]{};
		
		\vertex[fill] (w00) at (3,1) [label=90:$x_0$]{};
		\vertex[fill] (w01) at (3.866,0.5) [label=15:$y_0$]{};
		\vertex[fill] (w02) at (3.866,-0.5) [label=345:$x_1$]{};
		\vertex[fill] (w03) at (3,-1) [label=270:$y_1$]{};
		\vertex[fill] (w04) at (2.134,-0.5) [label=195:$x_2$]{};
		\vertex[fill] (w05) at (2.134,0.5) [label=165:$y_2$]{};
			
		\path 
			(v00) edge (v01)
			(v02) edge (v03)
			(v04) edge (v05)
			
			(w00) edge (w01)
			(w02) edge (w03)
			(w04) edge (w05)
		;

		\vertex[fill] (v10) at (8,1) [label=90:$v_0$]{};
		\vertex[fill] (v11) at (8.866,0.5) [label=15:$w_0$]{};
		\vertex[fill] (v12) at (8.866,-0.5) [label=345:$v_1$]{};
		\vertex[fill] (v13) at (8,-1) [label=270:$w_1$]{};
		\vertex[fill] (v14) at (7.134,-0.5) [label=195:$v_2$]{};
		\vertex[fill] (v15) at (7.134,0.5) [label=165:$w_2$]{};
		
		\vertex[fill] (w10) at (11,1) [label=90:$x_0$]{};
		\vertex[fill] (w11) at (11.866,0.5) [label=15:$y_0$]{};
		\vertex[fill] (w12) at (11.866,-0.5) [label=345:$x_1$]{};
		\vertex[fill] (w13) at (11,-1) [label=270:$y_1$]{};
		\vertex[fill] (w14) at (10.134,-0.5) [label=195:$x_2$]{};
		\vertex[fill] (w15) at (10.134,0.5) [label=165:$y_2$]{};
			
		\path 
			(v11) edge (v12)
			(v13) edge (v14)
			(v15) edge (v10)
			
			(w11) edge (w12)
			(w13) edge (w14)
			(w15) edge (w10)
		;

    \draw[draw=black] (0,-2) -- (11,-2);

		\vertex[fill] (v20) at (0,-3) [label=90:$v_0$]{};
		\vertex[fill] (v21) at (0.866,-3.5) [label=15:$w_0$]{};
		\vertex[fill] (v22) at (0.866,-4.5) [label=345:$v_1$]{};
		\vertex[fill] (v23) at (0,-5) [label=270:$w_1$]{};
		\vertex[fill] (v24) at (-0.866,-4.5) [label=195:$v_2$]{};
		\vertex[fill] (v25) at (-0.866,-3.5) [label=165:$w_2$]{};
		
		\vertex[fill] (w20) at (3,-3) [label=90:$x_0$]{};
		\vertex[fill] (w21) at (3.866,-3.5) [label=15:$y_0$]{};
		\vertex[fill] (w22) at (3.866,-4.5) [label=345:$x_1$]{};
		\vertex[fill] (w23) at (3,-5) [label=270:$y_1$]{};
		\vertex[fill] (w24) at (2.134,-4.5) [label=195:$x_2$]{};
		\vertex[fill] (w25) at (2.134,-3.5) [label=165:$y_2$]{};
			
		\path 
			(v21) edge (v24)
			(v20) edge (v22)
			(v23) edge (v25)
			
			(w21) edge (w24)
			(w20) edge (w22)
			(w23) edge (w25)
		;

		\vertex[fill] (v30) at (8,-3) [label=90:$v_0$]{};
		\vertex[fill] (v31) at (8.866,-3.5) [label=15:$w_0$]{};
		\vertex[fill] (v32) at (8.866,-4.5) [label=345:$v_1$]{};
		\vertex[fill] (v33) at (8,-5) [label=270:$w_1$]{};
		\vertex[fill] (v34) at (7.134,-4.5) [label=195:$v_2$]{};
		\vertex[fill] (v35) at (7.134,-3.5) [label=165:$w_2$]{};
		
		\vertex[fill] (w30) at (11,-3) [label=90:$x_0$]{};
		\vertex[fill] (w31) at (11.866,-3.5) [label=15:$y_0$]{};
		\vertex[fill] (w32) at (11.866,-4.5) [label=345:$x_1$]{};
		\vertex[fill] (w33) at (11,-5) [label=270:$y_1$]{};
		\vertex[fill] (w34) at (10.134,-4.5) [label=195:$x_2$]{};
		\vertex[fill] (w35) at (10.134,-3.5) [label=165:$y_2$]{};
			
		\path 
			(v32) edge (v35)
			(v30) edge (v34)
			(v31) edge (v33)
			
			(w32) edge (w35)
			(w30) edge (w34)
			(w31) edge (w33)
		;

    \draw[draw=black] (0,-6) -- (11,-6);

		\vertex[fill] (v40) at (0,-7) [label=90:$v_0$]{};
		\vertex[fill] (v41) at (0.866,-7.5) [label=15:$w_0$]{};
		\vertex[fill] (v42) at (0.866,-8.5) [label=345:$v_1$]{};
		\vertex[fill] (v43) at (0,-9) [label=270:$w_1$]{};
		\vertex[fill] (v44) at (-0.866,-8.5) [label=195:$v_2$]{};
		\vertex[fill] (v45) at (-0.866,-7.5) [label=165:$w_2$]{};
		
		\vertex[fill] (w40) at (3,-7) [label=90:$x_0$]{};
		\vertex[fill] (w41) at (3.866,-7.5) [label=15:$y_0$]{};
		\vertex[fill] (w42) at (3.866,-8.5) [label=345:$x_1$]{};
		\vertex[fill] (w43) at (3,-9) [label=270:$y_1$]{};
		\vertex[fill] (w44) at (2.134,-8.5) [label=195:$x_2$]{};
		\vertex[fill] (w45) at (2.134,-7.5) [label=165:$y_2$]{};
			
		\path 
			(v40) edge (v43)
			(v41) edge (v45)
			(v42) edge (v44)
			
			(w40) edge (w43)
			(w41) edge (w45)
			(w42) edge (w44)
		;
	\end{tikzpicture}
	\caption{A $1$-factorization of $K_{6}^2$, which is useful for constructing a UPB of size $12$ in $(\bb{C}^2)^{\otimes 8}$.}\label{fig:2complete_graphs_1fac}
\end{figure}
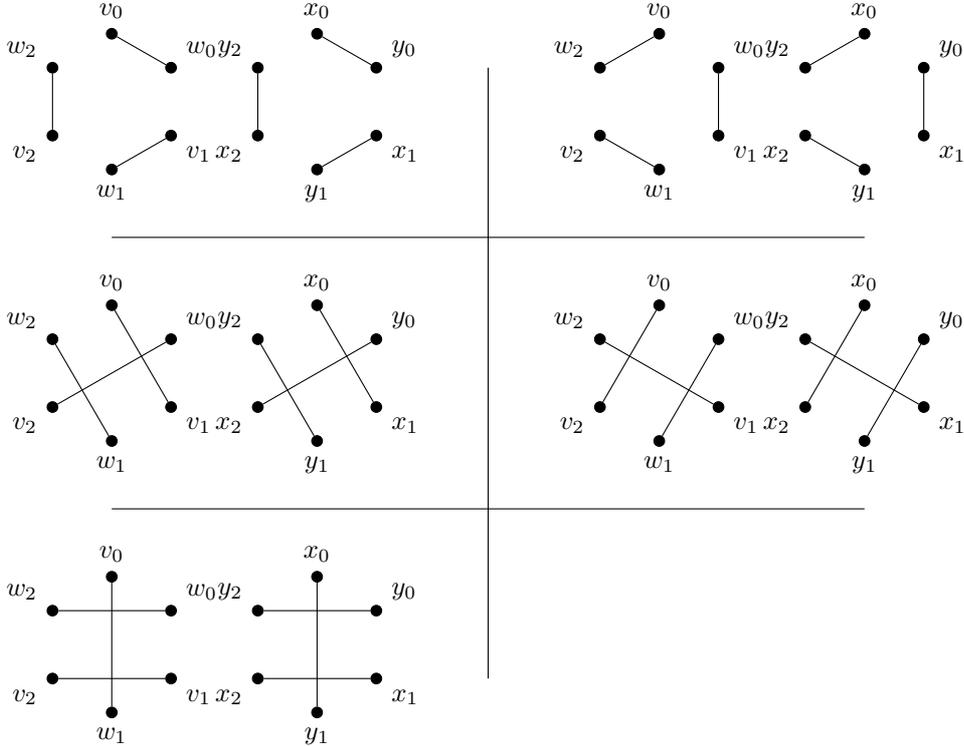

Since the union of the sets of edges present in all of the graphs considered so far is the complete graph $K_{4k+4}$, we know that the states in the set
\begin{align*}
	\cl{S} := \left\{ \bigotimes_{j=1}^{4k} \ket{v_i^{(j)}}, \bigotimes_{j=1}^{4k} \ket{w_i^{(j)}}, \bigotimes_{j=1}^{4k} \ket{x_i^{(j)}}, \bigotimes_{j=1}^{4k} \ket{y_i^{(j)}} : 0 \leq i \leq k \right\}
\end{align*}
are mutually orthogonal. To see why this set is unextendible, recall that any non-zero product state can be orthogonal to at most $2$ states on each of the first $3$ subsystems, and at most $1$ state on each of the remaining $4k-3$ subsystems. It follows that any nonzero product state can be orthogonal to at most $2\cdot 3 + 1 \cdot (4k-3) = 4k+3$ of these product states. Since no nonzero product state can be orthogonal to all $4k+4$ members of $\cl{S}$, it is unextendible, which completes the proof.
\end{proof}

\section{Proof of Minimality}\label{sec:minimality_proof}

We now turn our attention to the problem of proving that the UPBs constructed in Section~\ref{sec:construction} are the smallest possible. Because the main result of \cite{AL01} tells us that the minimum cardinality of a UPB in $(\bb{C}^2)^{\otimes 4k}$ is at least $4k+2$, we only have to prove that there is no UPB of cardinality $4k+2$ when $k \geq 2$ and no UPB of cardinality $4k+3$ when $k \geq 3$. While the proof that there is no UPB of cardinality $4k+2$ is relatively straightforward, the proof that there is no UPB of cardinality $4k+3$ is more involved and consists of many cases and sub-cases. We make use of a C script to solve some of the messier cases, while we solve the simpler cases by hand.

For the entirety of this section, we make use of \emph{partial orthogonality graphs}, which are the same as orthogonality graphs, except perhaps with some conditions unspecified. For example, in Figure~\ref{fig:p2_case1} the lack of lines indicating orthogonality between shaded regions does not signify that there are no regions orthogonal to each other, but rather that we just don't care \emph{which} regions are orthogonal to each other. Similarly, in Figure~\ref{fig:p3_case1} there are vertices that are drawn outside of any shaded region. This is intended to mean that we don't care what the shaded region involving that vertex looks like. In general, we only specify the pieces of the orthogonality graphs that are relevant for our proofs.

It will be convenient for us to let $P_1, \ldots, P_{4k}$ denote the $4k$ different parties. We also let $M_j$ denote the maximum number of vertices contained within a single shaded region on party $P_j$ (which is equal to the maximum number of states in the UPB that are equal to each other on party $P_j$), and let $C_{n,j}$ denote the number of distinct shaded regions containing exactly $n$ vertices on party $j$ (i.e., $C_{n,j}$ is the number of distinct group of exactly $n$ states in the UPB that are equal to each other on party $P_j$). For example, in Figure~\ref{fig:2dim_eg_compact}, if the graphs correspond to parties $P_1$, $P_2$ and $P_3$, then $M_1 = 4$, $M_2 = M_3 = 2$, $C_{3,1} = 1$, $C_{4,1} = 1$, $C_{1,2} = 1$, $C_{2,2} = 3$, $C_{1,3} = 5$, and $C_{2,3} = 1$.

\begin{lemma}\label{lem:upb_p2}
	There is no UPB in $(\bb{C}^2)^{\otimes 4k}$ of cardinality $4k+2$ when $k \geq 2$.
\end{lemma}
\begin{proof}
	Suppose for a contradiction that there exists a UPB of cardinality $4k+2$ in $(\bb{C}^2)^{\otimes 4k}$. If it were the case that $M_j \geq 3$ for some $j$, then we could find a product state that is orthogonal to the $3$ corresponding states on that party and to any $1$ of the product states on each of the remaining $4k-1$ parties, for a total of all $4k+2$ elements of the UPB, which violates unextendibility. Hence $M_j \leq 2$ for all $1 \leq j \leq 4k$. We now split into two cases.
	
	\noindent {\bf Case 1:} There is at most one party $P_j$ with $M_j = 2$.

	Between the $4k$ parties, there must be a total of $(4k+2)(4k+1)/2 = 8k^2 + 6k + 1$ edges in their orthogonality graphs. The $4k-1$ parties other than $P_j$ must be the disjoint union of $2k+1$ copies of $K_{1,1}$, for a total of at most $(4k-1)(2k+1) = 8k^2 + 2k - 1$ edges. The remaining party $P_j$ then needs at least $(8k^2 + 6k + 1) - (8k^2 + 2k - 1) = 4k + 2$ edges. It is easily seen, however, that the largest number of edges that the orthogonality graph of party $P_j$ can have is obtained when it is the disjoint union of $k$ copies of $K_{2,2}$ and one copy of $K_{1,1}$, which results in only $4k + 1$ edges, which gives the desired contradiction.
	
	\noindent {\bf Case 2:} There are two (or more) parties $P_i \neq P_j$ with $M_i = M_j = 2$.

	It is not difficult to see that $C_{2,\ell} \in \{0,2\}$ for all $\ell$ or else either Lemma~\ref{lem:all_pos} or unextendibility is violated. Furthermore, it is not difficult to see that the unique (up to repositioning vertices and parties) way to have $C_{2,\ell} = 2$ for $3$ distinct values of $\ell$ is given in Figure~\ref{fig:p2_case1}, and there is no way to have $C_{2,\ell}$ for a fourth value of $\ell$ without violating unextendibility. A simple calculation reveals that the maximum number of edges that can be obtained from the orthogonality graphs of these $3$ parties is $(2k + 3) + 2(2k+2) = 6k + 7$. The orthogonality graphs of the remaining $4k-3$ parties are the disjoint union of $2k+1$ copies of $K_{1,1}$, so they each have $2k+1$ edges. Thus the total number of edges among the orthogonality graphs of all $4k$ parties is at most $(6k+7) + (4k-3)(2k+1) = 8k^2 + 4k + 4$. This quantity is smaller than the $8k^2 + 6k + 1$ required edges when $k \geq 2$, which gives the desired contradiction.
	
\begin{figure}[htb]
	\centering
	\begin{tikzpicture}[x=1.4cm, y=1.4cm, label distance=0cm]     
    \filldraw[fill=lightgray,line width=0.34cm,line join=round,draw=black] (-0.541,0.841) -- (-0.756,-0.655) -- (-0.6485,0.093) -- cycle;
    \filldraw[fill=lightgray,line width=0.30cm,line join=round,draw=lightgray] (-0.541,0.841) -- (-0.756,-0.655) -- (-0.6485,0.093) -- cycle;
    \filldraw[fill=lightgray,line width=0.34cm,line join=round,draw=black] (-0.910,0.415) -- (-0.990,-0.142) -- (-0.9500,0.1365) -- cycle;
    \filldraw[fill=lightgray,line width=0.30cm,line join=round,draw=lightgray] (-0.910,0.415) -- (-0.990,-0.142) -- (-0.9500,0.1365) -- cycle;

    \draw[line width=0.02cm,draw=black,fill=lightgray] (0.541,0.841) circle (0.2cm);
    \draw[line width=0.02cm,draw=black,fill=lightgray] (0,1) circle (0.2cm);
    \draw[line width=0.02cm,draw=black,fill=lightgray] (0.282,-0.959) circle (0.2cm);
    \draw[line width=0.02cm,draw=black,fill=lightgray] (-0.282,-0.959) circle (0.2cm);

		\vertex[fill] (v00) at (0,1) []{};
		\vertex[fill] (v01) at (-0.541,0.841) []{};
		\vertex[fill] (v02) at (-0.910,0.415) []{};
		\vertex[fill] (v03) at (-0.990,-0.142) []{};
		\vertex[fill] (v04) at (-0.756,-0.655) []{};
		\vertex[fill] (v05) at (-0.282,-0.959) []{};
		\vertex[fill] (v06) at (0.282,-0.959) []{};
		\vertex[fill] (v010) at (0.541,0.841) []{};
		\draw[thin,fill] (0.837,0.547) circle(0.02);
		\draw[thin,fill] (0.910,0.415) circle(0.02);
		\draw[thin,fill] (0.961,0.275) circle(0.02);
		\draw[thin,fill] (0.845,-0.535) circle(0.02);
		\draw[thin,fill] (0.756,-0.655) circle(0.02);
		\draw[thin,fill] (0.649,-0.760) circle(0.02);

    \filldraw[fill=lightgray,line width=0.34cm,line join=round,draw=black] (2.9590,0.8410) -- (2.5900,0.4150) -- (2.7745,0.6280) -- cycle;
    \filldraw[fill=lightgray,line width=0.30cm,line join=round,draw=lightgray] (2.9590,0.8410) -- (2.5900,0.4150) -- (2.7745,0.6280) -- cycle;
    \filldraw[fill=lightgray,line width=0.34cm,line join=round,draw=black] (2.5100,-0.1420) -- (2.7440,-0.6550) -- (2.6270,-0.3985) -- cycle;
    \filldraw[fill=lightgray,line width=0.30cm,line join=round,draw=lightgray] (2.5100,-0.1420) -- (2.7440,-0.6550) -- (2.6270,-0.3985) -- cycle;

    \draw[line width=0.02cm,draw=black,fill=lightgray] (3.2180,-0.9590) circle (0.2cm);
    \draw[line width=0.02cm,draw=black,fill=lightgray] (3.782,-0.959) circle (0.2cm);
    \draw[line width=0.02cm,draw=black,fill=lightgray] (3.5,1) circle (0.2cm);
    \draw[line width=0.02cm,draw=black,fill=lightgray] (4.041,0.841) circle (0.2cm);

		\vertex[fill] (v10) at (3.5,1) []{};
		\vertex[fill] (v11) at (2.9590,0.8410) []{};
		\vertex[fill] (v12) at (2.5900,0.4150) []{};
		\vertex[fill] (v13) at (2.5100,-0.1420) []{};
		\vertex[fill] (v14) at (2.7440,-0.6550) []{};
		\vertex[fill] (v15) at (3.2180,-0.9590) []{};
		\vertex[fill] (v16) at (3.782,-0.959) []{};
		\vertex[fill] (v110) at (4.041,0.841) []{};
		\draw[thin,fill] (4.337,0.547) circle(0.02);
		\draw[thin,fill] (4.410,0.415) circle(0.02);
		\draw[thin,fill] (4.461,0.275) circle(0.02);
		\draw[thin,fill] (4.345,-0.535) circle(0.02);
		\draw[thin,fill] (4.256,-0.655) circle(0.02);
		\draw[thin,fill] (4.149,-0.760) circle(0.02);

    \filldraw[fill=lightgray,line width=0.34cm,line join=round,draw=black] (6.31,-0.1420) -- (6.0900,0.4150) -- (6.3,-0.1420) -- (6.2440,-0.6550) -- cycle;
    \filldraw[fill=lightgray,line width=0.30cm,line join=round,draw=lightgray] (6.31,-0.1420) -- (6.0900,0.4150) -- (6.3,-0.1420) -- (6.2440,-0.6550) -- cycle;
    \filldraw[fill=lightgray,line width=0.34cm,line join=round,draw=black] (5.7419,0.5691) -- (6.4590,0.8410) -- (5.7419,0.5691) -- (6.0100,-0.1420) -- cycle;
    \filldraw[fill=lightgray,line width=0.30cm,line join=round,draw=lightgray] (5.7419,0.5691) -- (6.4590,0.8410) -- (5.7419,0.5691) -- (6.0100,-0.1420) -- cycle;

    \draw[line width=0.02cm,draw=black,fill=lightgray] (6.7180,-0.9590) circle (0.2cm);
    \draw[line width=0.02cm,draw=black,fill=lightgray] (7.282,-0.959) circle (0.2cm);
    \draw[line width=0.02cm,draw=black,fill=lightgray] (7,1) circle (0.2cm);
    \draw[line width=0.02cm,draw=black,fill=lightgray] (7.541,0.841) circle (0.2cm);

		\vertex[fill] (v20) at (7,1) []{};
		\vertex[fill] (v21) at (6.4590,0.8410) []{};
		\vertex[fill] (v22) at (6.0900,0.4150) []{};
		\vertex[fill] (v23) at (6.0100,-0.1420) []{};
		\vertex[fill] (v24) at (6.2440,-0.6550) []{};
		\vertex[fill] (v25) at (6.7180,-0.9590) []{};
		\vertex[fill] (v26) at (7.282,-0.959) []{};
		\vertex[fill] (v210) at (7.541,0.841) []{};
		\draw[thin,fill] (7.837,0.547) circle(0.02);
		\draw[thin,fill] (7.910,0.415) circle(0.02);
		\draw[thin,fill] (7.961,0.275) circle(0.02);
		\draw[thin,fill] (7.845,-0.535) circle(0.02);
		\draw[thin,fill] (7.756,-0.655) circle(0.02);
		\draw[thin,fill] (7.649,-0.760) circle(0.02);
	\end{tikzpicture}
	\caption{Partial orthogonality graphs of three parties that each have two sets of two equal states, used in the proof of case~2 of Lemma~\ref{lem:upb_p2}. There is no way to add another pair of equal states on any party without violating unextendibility.}\label{fig:p2_case1}
\end{figure}
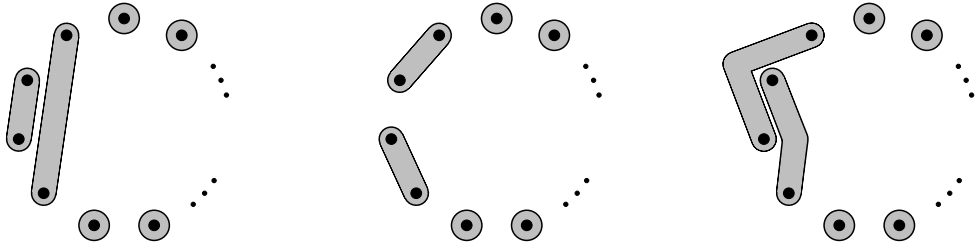
\end{proof}

Note that the hypothesis of Lemma~\ref{lem:upb_p2} that $k \geq 2$ really is required, since we have $8k^2 + 4k + 4 \geq 8k^2 + 6k + 1$ in case~2 of the proof of the lemma when $k = 1$, so it may be possible to fit all of the required edges into the orthogonality graphs. Indeed, it was shown in \cite{Fen06} that a UPB consisting of $4k+2$ states in $(\bb{C}^2)^{\otimes 4k}$ exists in the $k = 1$ case.

We now turn our attention to proving that there is no UPB of cardinality $4k+3$ when $k \geq 3$. The idea and techniques used in the proof of this statement are quite similar to the $4k+2$ case, but there are more cases to consider.

\begin{lemma}\label{lem:upb_p3}
	There is no UPB in $(\bb{C}^2)^{\otimes 4k}$ of cardinality $4k+3$ when $k \geq 3$.
\end{lemma}
\begin{proof}
	Suppose for a contradiction that there exists a UPB of cardinality $4k+3$ in $(\bb{C}^2)^{\otimes 4k}$. If there exists $1 \leq j \leq p$ such that $M_j \geq 4$, then we can find a product state that is orthogonal to at least $4$ corresponding states on party $P_j$ and to $1$ of the product states on each of the remaining $4k-1$ parties, for a total of $4k+3$ elements of the UPB, which violates unextendibility. Hence $M_j \leq 3$ for all $j$. Furthermore, this same argument shows that if there exists $i \geq 1$ such that we can choose a single shaded region on each of $i$ parties so that together they contain at least $i + 3$ vertices, then unextendibility will be violated. Finally, note that since $4k+3$ is odd, Lemma~\ref{lem:all_pos} implies that $M_j \geq 2$ for all $j$.
	
	We now split into $4$ cases, depending on the value of ${\rm max}_j\{C_{3,j}\}$ (i.e., the maximum number of sets of $3$ equal states on any party).

	\noindent {\bf Case 1:} ${\rm max}_j\{C_{3,j}\} \geq 3$.
	
	Because $M_j \geq 2$ for all $j$, it easily follows that we can find shaded regions on two parties that contain $3 + 2 = 5$ distinct vertices, which contradicts unextendibility.
	
	\noindent {\bf Case 2:} ${\rm max}_j\{C_{3,j}\} = 2$.
	
	Suppose without loss of generality that party $P_1$ is such that $C_{3,1} = 2$. Unextendibility immediately implies that $C_{3,j} = 0$ for $j \geq 2$. Since there are $4k-3$ left over vertices on party $P_1$, which is odd, there must be a copy of $K_{2,1}$ on this party, as in Figure~\ref{fig:p3_case1}. Since $v_1$ is connected to only one other state on party $P_1$, it must be connected to $2$ states on each of $2$ other parties. These sets of $2$ vertices must be disjoint and must each contain one of $v_2,v_3,v_4$ and one of $v_5,v_6,v_7$. Thus parties $P_2$ and $P_3$, without loss of generality, are as in Figure~\ref{fig:p3_case1}, which clearly implies extendibility and rules out this case.
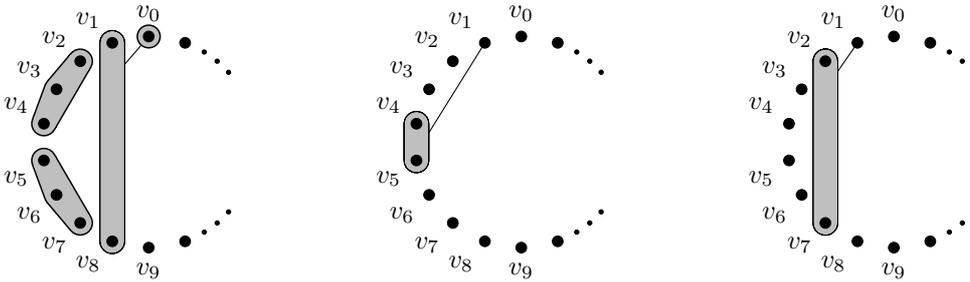
\begin{figure}[htb]
	\centering
	\begin{tikzpicture}[x=1.4cm, y=1.4cm, label distance=0cm]     
    \draw[draw=black] (0,1) -- (-0.342,0.6);

    \filldraw[fill=lightgray,line width=0.34cm,line join=round,draw=black] (-0.342,0.940) -- (-0.342,-0.940) -- (-0.3420,0) -- cycle;
    \filldraw[fill=lightgray,line width=0.30cm,line join=round,draw=lightgray] (-0.342,0.940) -- (-0.342,-0.940) -- (-0.3420,0) -- cycle;
    \filldraw[fill=lightgray,line width=0.34cm,line join=round,draw=black] (-0.643,0.766) -- (-0.866,0.500) -- (-0.985,0.174) -- cycle;
    \filldraw[fill=lightgray,line width=0.30cm,line join=round,draw=lightgray] (-0.643,0.766) -- (-0.866,0.500) -- (-0.985,0.174) -- cycle;
	\filldraw[fill=lightgray,line width=0.34cm,line join=round,draw=black] (-0.985,-0.174) -- (-0.866,-0.500) -- (-0.643,-0.766) -- cycle;
    \filldraw[fill=lightgray,line width=0.30cm,line join=round,draw=lightgray] (-0.985,-0.174) -- (-0.866,-0.500) -- (-0.643,-0.766) -- cycle;

    \draw[line width=0.02cm,draw=black,fill=lightgray] (0,1) circle (0.15cm);

		\vertex[fill] (v000) at (0,1) [label=90:$v_{0}$]{};
		\vertex[fill] (v001) at (-0.342,0.940) [label=110:$v_{1}$]{};
		\vertex[fill] (v002) at (-0.643,0.766) [label=130:$v_{2}$]{};
		\vertex[fill] (v003) at (-0.866,0.500) [label=150:$v_{3}$]{};
		\vertex[fill] (v004) at (-0.985,0.174) [label=170:$v_{4}$]{};
		\vertex[fill] (v005) at (-0.985,-0.174) [label=190:$v_{5}$]{};
		\vertex[fill] (v006) at (-0.866,-0.500) [label=210:$v_{6}$]{};
		\vertex[fill] (v007) at (-0.643,-0.766) [label=230:$v_{7}$]{};
		\vertex[fill] (v008) at (-0.342,-0.940) [label=250:$v_{8}$]{};
		\vertex[fill] (v009) at (0,-1) [label=270:$v_{9}$]{};
		\vertex[fill] (v010) at (0.342,-0.940) []{};
		\vertex[fill] (v017) at (0.342,0.940) []{};
		\draw[thin,fill] (0.750,0.661) circle(0.02);
		\draw[thin,fill] (0.643,0.766) circle(0.02);
		\draw[thin,fill] (0.521,0.853) circle(0.02);
		\draw[thin,fill] (0.521,-0.853) circle(0.02);
		\draw[thin,fill] (0.643,-0.766) circle(0.02);
		\draw[thin,fill] (0.750,-0.661) circle(0.02);

    \draw[draw=black] (3.158,0.940) -- (2.515,-0.1);

    \filldraw[fill=lightgray,line width=0.34cm,line join=round,draw=black] (2.515,0.174) -- (2.515,-0.174) -- (2.5150,0) -- cycle;
    \filldraw[fill=lightgray,line width=0.30cm,line join=round,draw=lightgray] (2.515,0.174) -- (2.515,-0.174) -- (2.5150,0) -- cycle;

		\vertex[fill] (v100) at (3.5,1) [label=90:$v_{0}$]{};
		\vertex[fill] (v101) at (3.158,0.940) [label=110:$v_{1}$]{};
		\vertex[fill] (v102) at (2.857,0.766) [label=130:$v_{2}$]{};
		\vertex[fill] (v103) at (2.634,0.500) [label=150:$v_{3}$]{};
		\vertex[fill] (v104) at (2.515,0.174) [label=170:$v_{4}$]{};
		\vertex[fill] (v105) at (2.515,-0.174) [label=190:$v_{5}$]{};
		\vertex[fill] (v106) at (2.634,-0.500) [label=210:$v_{6}$]{};
		\vertex[fill] (v107) at (2.857,-0.766) [label=230:$v_{7}$]{};
		\vertex[fill] (v108) at (3.158,-0.940) [label=250:$v_{8}$]{};
		\vertex[fill] (v109) at (3.5,-1) [label=270:$v_{9}$]{};
		\vertex[fill] (v110) at (3.842,-0.940) []{};
		\vertex[fill] (v117) at (3.842,0.940) []{};
		\draw[thin,fill] (4.250,0.661) circle(0.02);
		\draw[thin,fill] (4.143,0.766) circle(0.02);
		\draw[thin,fill] (4.021,0.853) circle(0.02);
		\draw[thin,fill] (4.021,-0.853) circle(0.02);
		\draw[thin,fill] (4.143,-0.766) circle(0.02);
		\draw[thin,fill] (4.250,-0.661) circle(0.02);

    \draw[draw=black] (6.658,0.940) -- (6.357,0.5);

    \filldraw[fill=lightgray,line width=0.34cm,line join=round,draw=black] (6.357,0.766) -- (6.357,-0.766) -- (6.357,0) -- cycle;
    \filldraw[fill=lightgray,line width=0.30cm,line join=round,draw=lightgray] (6.357,0.766) -- (6.357,-0.766) -- (6.357,0) -- cycle;

		\vertex[fill] (v200) at (7,1) [label=90:$v_{0}$]{};
		\vertex[fill] (v201) at (6.658,0.940) [label=110:$v_{1}$]{};
		\vertex[fill] (v202) at (6.357,0.766) [label=130:$v_{2}$]{};
		\vertex[fill] (v203) at (6.134,0.500) [label=150:$v_{3}$]{};
		\vertex[fill] (v204) at (6.015,0.174) [label=170:$v_{4}$]{};
		\vertex[fill] (v205) at (6.015,-0.174) [label=190:$v_{5}$]{};
		\vertex[fill] (v206) at (6.134,-0.500) [label=210:$v_{6}$]{};
		\vertex[fill] (v207) at (6.357,-0.766) [label=230:$v_{7}$]{};
		\vertex[fill] (v208) at (6.658,-0.940) [label=250:$v_{8}$]{};
		\vertex[fill] (v209) at (7,-1) [label=270:$v_{9}$]{};
		\vertex[fill] (v210) at (7.342,-0.940) []{};
		\vertex[fill] (v217) at (7.342,0.940) []{};
		\draw[thin,fill] (7.750,0.661) circle(0.02);
		\draw[thin,fill] (7.643,0.766) circle(0.02);
		\draw[thin,fill] (7.521,0.853) circle(0.02);
		\draw[thin,fill] (7.521,-0.853) circle(0.02);
		\draw[thin,fill] (7.643,-0.766) circle(0.02);
		\draw[thin,fill] (7.750,-0.661) circle(0.02);
	\end{tikzpicture}
	\caption{The (essentially unique) partial orthogonality graphs of parties $P_1$ (left), $P_2$ (center) and $P_3$ (right) in case~2 of Lemma~\ref{lem:upb_p3}. Such a product basis is necessarily extendible, as we can find a product state that is orthogonal to the states corresponding to $v_1$ and $v_8$ on party $P_1$, $v_4$ and $v_5$ on party $P_2$, $v_2$ and $v_7$ on party $P_3$, and one of the $4k-3$ remaining states on each of the remaining $4k-3$ parties.}\label{fig:p3_case1}
\end{figure}
	
	\noindent {\bf Case 3:} ${\rm max}_j\{C_{3,j}\} = 0$.
	
	Since $M_j = 2$ for all $j$, simple parity arguments show that $C_{2,j} \in \{1,3,5,\ldots\}$ for every $j$. We now split into two sub-cases, depending on the value of ${\rm max}_j \{C_{2,j}\}$ (i.e., the maximum number of sets of $2$ equal states on any party).
	
	\noindent {\bf Case 3(a):} ${\rm max}_j \{C_{2,j}\} \geq 5$.
	
	Suppose that party $P_1$ has $C_{2,1} = 5$. We first argue that there must be at least one other party $P_2$ with $C_{2,2} \geq 3$. To see this, suppose the contrary -- suppose that $C_{2,j} = 1$ for all $j \geq 2$. Then each of these $4k-1$ parties contributes at most $2k + 2$ edges to the orthogonality graph, for a total of $(4k-1)(2k+2) = 8k^2 + 6k - 2$ edges. The party $P_1$ contributes no more than $4k + 2$ edges, for a total of $8k^2 + 10k$ edges among all $4k$ parties. However, the complete graph on $4k+3$ vertices has $(4k+3)(4k+2)/2 = 8k^2 + 10k + 3$ edges, so there are at least $3$ pairs of non-orthogonal product states in our set, which contradicts the assumption that we are working with a UPB.
	
	We now pick an arbitrary party $P_3 \neq P_1,P_2$. Because $C_{2,3} \geq 1$, we are now able to choose one shaded region on each of parties $P_1,P_2,P_3$ such that $6$ vertices are contained within these regions, which shows that unextendibility is violated. To this end, we choose any shaded region on party $P_3$ that contains two vertices, then we pick any shaded region on party $P_2$ that is disjoint from the two vertices we chose on party $P_3$, and finally we choose any shaded region on party $P_1$ that is disjoint from all four of the previously-chosen vertices (see Figure~\ref{fig:p3_case2a}).

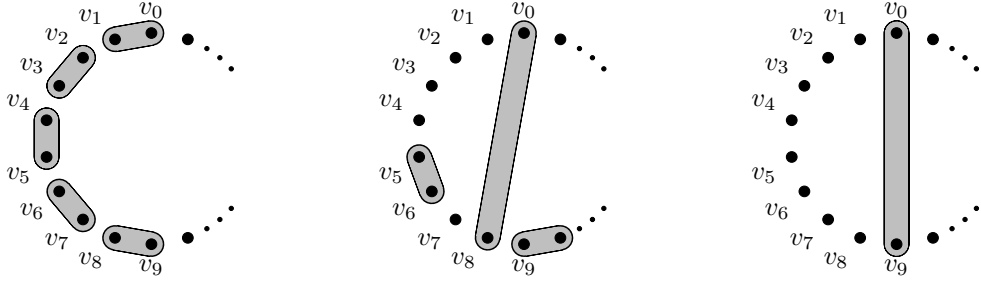
\begin{figure}[htb]
	\centering
	\begin{tikzpicture}[x=1.4cm, y=1.4cm, label distance=0cm]     
    \filldraw[fill=lightgray,line width=0.34cm,line join=round,draw=black] (0,1) -- (-0.342,0.940) -- (-0.1710,0.9700) -- cycle;
    \filldraw[fill=lightgray,line width=0.30cm,line join=round,draw=lightgray] (0,1) -- (-0.342,0.940) -- (-0.1710,0.9700) -- cycle;
    \filldraw[fill=lightgray,line width=0.34cm,line join=round,draw=black] (-0.643,0.766) -- (-0.866,0.500) -- (-0.7545,0.6330) -- cycle;
    \filldraw[fill=lightgray,line width=0.30cm,line join=round,draw=lightgray] (-0.643,0.766) -- (-0.866,0.500) -- (-0.7545,0.6330) -- cycle;
	\filldraw[fill=lightgray,line width=0.34cm,line join=round,draw=black] (-0.985,0.174) -- (-0.985,-0.174) -- (-0.9850,0) -- cycle;
    \filldraw[fill=lightgray,line width=0.30cm,line join=round,draw=lightgray] (-0.985,0.174) -- (-0.985,-0.174) -- (-0.9850,0) -- cycle;
	\filldraw[fill=lightgray,line width=0.34cm,line join=round,draw=black] (-0.866,-0.500) -- (-0.643,-0.766) -- (-0.7545,-0.6330) -- cycle;
    \filldraw[fill=lightgray,line width=0.30cm,line join=round,draw=lightgray] (-0.866,-0.500) -- (-0.643,-0.766) -- (-0.7545,-0.6330) -- cycle;
	\filldraw[fill=lightgray,line width=0.34cm,line join=round,draw=black] (-0.342,-0.940) -- (0,-1) -- (-0.1710,-0.9700) -- cycle;
    \filldraw[fill=lightgray,line width=0.30cm,line join=round,draw=lightgray] (-0.342,-0.940) -- (0,-1) -- (-0.1710,-0.9700) -- cycle;

		\vertex[fill] (v000) at (0,1) [label=90:$v_{0}$]{};
		\vertex[fill] (v001) at (-0.342,0.940) [label=110:$v_{1}$]{};
		\vertex[fill] (v002) at (-0.643,0.766) [label=130:$v_{2}$]{};
		\vertex[fill] (v003) at (-0.866,0.500) [label=150:$v_{3}$]{};
		\vertex[fill] (v004) at (-0.985,0.174) [label=170:$v_{4}$]{};
		\vertex[fill] (v005) at (-0.985,-0.174) [label=190:$v_{5}$]{};
		\vertex[fill] (v006) at (-0.866,-0.500) [label=210:$v_{6}$]{};
		\vertex[fill] (v007) at (-0.643,-0.766) [label=230:$v_{7}$]{};
		\vertex[fill] (v008) at (-0.342,-0.940) [label=250:$v_{8}$]{};
		\vertex[fill] (v009) at (0,-1) [label=270:$v_{9}$]{};
		\vertex[fill] (v010) at (0.342,-0.940) []{};
		\vertex[fill] (v017) at (0.342,0.940) []{};
		\draw[thin,fill] (0.750,0.661) circle(0.02);
		\draw[thin,fill] (0.643,0.766) circle(0.02);
		\draw[thin,fill] (0.521,0.853) circle(0.02);
		\draw[thin,fill] (0.521,-0.853) circle(0.02);
		\draw[thin,fill] (0.643,-0.766) circle(0.02);
		\draw[thin,fill] (0.750,-0.661) circle(0.02);

    \filldraw[fill=lightgray,line width=0.34cm,line join=round,draw=black] (3.5,1) -- (3.158,-0.940) -- (3.3290,0.0300) -- cycle;
    \filldraw[fill=lightgray,line width=0.30cm,line join=round,draw=lightgray] (3.5,1) -- (3.158,-0.940) -- (3.3290,0.0300) -- cycle;
    \filldraw[fill=lightgray,line width=0.34cm,line join=round,draw=black] (2.515,-0.174) -- (2.634,-0.500) -- (2.5745,-0.3370) -- cycle;
    \filldraw[fill=lightgray,line width=0.30cm,line join=round,draw=lightgray] (2.515,-0.174) -- (2.634,-0.500) -- (2.5745,-0.3370) -- cycle;
    \filldraw[fill=lightgray,line width=0.34cm,line join=round,draw=black] (3.5,-1) -- (3.842,-0.940) -- (3.6710,-0.9700) -- cycle;
    \filldraw[fill=lightgray,line width=0.30cm,line join=round,draw=lightgray] (3.5,-1) -- (3.842,-0.940) -- (3.6710,-0.9700) -- cycle;

		\vertex[fill] (v100) at (3.5,1) [label=90:$v_{0}$]{};
		\vertex[fill] (v101) at (3.158,0.940) [label=110:$v_{1}$]{};
		\vertex[fill] (v102) at (2.857,0.766) [label=130:$v_{2}$]{};
		\vertex[fill] (v103) at (2.634,0.500) [label=150:$v_{3}$]{};
		\vertex[fill] (v104) at (2.515,0.174) [label=170:$v_{4}$]{};
		\vertex[fill] (v105) at (2.515,-0.174) [label=190:$v_{5}$]{};
		\vertex[fill] (v106) at (2.634,-0.500) [label=210:$v_{6}$]{};
		\vertex[fill] (v107) at (2.857,-0.766) [label=230:$v_{7}$]{};
		\vertex[fill] (v108) at (3.158,-0.940) [label=250:$v_{8}$]{};
		\vertex[fill] (v109) at (3.5,-1) [label=270:$v_{9}$]{};
		\vertex[fill] (v110) at (3.842,-0.940) []{};
		\vertex[fill] (v117) at (3.842,0.940) []{};
		\draw[thin,fill] (4.250,0.661) circle(0.02);
		\draw[thin,fill] (4.143,0.766) circle(0.02);
		\draw[thin,fill] (4.021,0.853) circle(0.02);
		\draw[thin,fill] (4.021,-0.853) circle(0.02);
		\draw[thin,fill] (4.143,-0.766) circle(0.02);
		\draw[thin,fill] (4.250,-0.661) circle(0.02);

    \filldraw[fill=lightgray,line width=0.34cm,line join=round,draw=black] (7,1) -- (7,-1) -- (7,0) -- cycle;
    \filldraw[fill=lightgray,line width=0.30cm,line join=round,draw=lightgray] (7,1) -- (7,-1) -- (7,0) -- cycle;

		\vertex[fill] (v200) at (7,1) [label=90:$v_{0}$]{};
		\vertex[fill] (v201) at (6.658,0.940) [label=110:$v_{1}$]{};
		\vertex[fill] (v202) at (6.357,0.766) [label=130:$v_{2}$]{};
		\vertex[fill] (v203) at (6.134,0.500) [label=150:$v_{3}$]{};
		\vertex[fill] (v204) at (6.015,0.174) [label=170:$v_{4}$]{};
		\vertex[fill] (v205) at (6.015,-0.174) [label=190:$v_{5}$]{};
		\vertex[fill] (v206) at (6.134,-0.500) [label=210:$v_{6}$]{};
		\vertex[fill] (v207) at (6.357,-0.766) [label=230:$v_{7}$]{};
		\vertex[fill] (v208) at (6.658,-0.940) [label=250:$v_{8}$]{};
		\vertex[fill] (v209) at (7,-1) [label=270:$v_{9}$]{};
		\vertex[fill] (v210) at (7.342,-0.940) []{};
		\vertex[fill] (v217) at (7.342,0.940) []{};
		\draw[thin,fill] (7.750,0.661) circle(0.02);
		\draw[thin,fill] (7.643,0.766) circle(0.02);
		\draw[thin,fill] (7.521,0.853) circle(0.02);
		\draw[thin,fill] (7.521,-0.853) circle(0.02);
		\draw[thin,fill] (7.643,-0.766) circle(0.02);
		\draw[thin,fill] (7.750,-0.661) circle(0.02);
	\end{tikzpicture}
	\caption{An example of a partial orthogonality graph in case~3(a) of Lemma~\ref{lem:upb_p3}. Such a product basis is necessarily extendible, as we can choose the shaded region containing $v_0$ and $v_9$ on party $P_3$, the disjoint shaded region (i.e., the one containing $v_5$ and $v_6$) on party $P_2$, and the disjoint shaded region (i.e., the one containing $v_2$ and $v_3$) on party $P_1$, for a total of $6$ vertices on $3$ parties.}\label{fig:p3_case2a}
\end{figure}

	\noindent {\bf Case 3(b):} ${\rm max}_j \{C_{2,j}\} \leq 3$.

	We begin by noting that the brute-force computer search shows that there can be no more than $4$ distinct parties $P_j$ for which $C_{2,j} \geq 3$ \cite{JohQubitUPBCode}. Each of these four parties has at most $2k + 4$ edges in its orthogonality graph, and each of the remaining $4k-4$ parties has at most $2k + 2$ edges on its orthogonality graph, for a total of at most $4(2k+4) + (4k-4)(2k+2) = 8k^2 + 8k + 8$ edges. The complete graph on $4k+3$ vertices has $(4k+3)(4k+2)/2 = 8k^2 + 10k + 3$ edges, so when $k \geq 3$ there are not enough edges in the orthogonality graph, so the set of states does not form a product basis, which contradicts our assumption that we are working with a UPB. Note that this is the case in which the UPB of Lemma~\ref{lem:8_min_construct} arises in the $k = 2$ case, so the fact that we require $k \geq 3$ here is not surprising.

\noindent {\bf Case 4:} ${\rm max}_j\{C_{3,j}\} = 1$.
	
	By parity arguments, we see that every party $P_j$ with $C_{3,j} = 1$ must also have $C_{2,j} \in \{1,3,5,\dots\}$. Furthermore, if there exist two (or more) parties $P_1, P_2$ such that $M_1 = M_2 = 3$, then unextendibility is violated unless $C_{2,j} = 1$ whenever $M_j = 3$.
	
\noindent {\bf Case 4(a):} There exist three (or more) parties $P_1, P_2, P_3$ such that $M_1 = M_2 = M_3 = 3$.

Because there must exist a shaded region containing exactly $2$ vertices on each party $P_1$, $P_2$, $P_3$, it is easily verified that the only possible configuration of shaded regions on those parties (up to repositioning vertices and parties) that doesn't break unextendibility is the one depicted in Figure~\ref{fig:case_4a}.
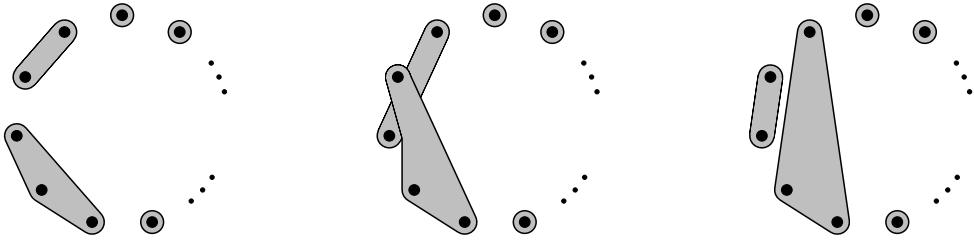
\begin{figure}[htb]
	\centering
	\begin{tikzpicture}[x=1.4cm, y=1.4cm, label distance=0cm]     
    \filldraw[fill=lightgray,line width=0.34cm,line join=round,draw=black] (-0.990,-0.142) -- (-0.756,-0.655) -- (-0.282,-0.959) -- cycle;
    \filldraw[fill=lightgray,line width=0.30cm,line join=round,draw=lightgray] (-0.990,-0.142) -- (-0.756,-0.655) -- (-0.282,-0.959) -- cycle;
    \filldraw[fill=lightgray,line width=0.34cm,line join=round,draw=black] (-0.541,0.841) -- (-0.910,0.415) -- (-0.7255,0.6280) -- cycle;
    \filldraw[fill=lightgray,line width=0.30cm,line join=round,draw=lightgray] (-0.541,0.841) -- (-0.910,0.415) -- (-0.7255,0.6280) -- cycle;

    \draw[line width=0.02cm,draw=black,fill=lightgray] (0.541,0.841) circle (0.15cm);
    \draw[line width=0.02cm,draw=black,fill=lightgray] (0,1) circle (0.15cm);
    \draw[line width=0.02cm,draw=black,fill=lightgray] (0.282,-0.959) circle (0.15cm);

		\vertex[fill] (v00) at (0,1) []{};
		\vertex[fill] (v01) at (-0.541,0.841) []{};
		\vertex[fill] (v02) at (-0.910,0.415) []{};
		\vertex[fill] (v03) at (-0.990,-0.142) []{};
		\vertex[fill] (v04) at (-0.756,-0.655) []{};
		\vertex[fill] (v05) at (-0.282,-0.959) []{};
		\vertex[fill] (v06) at (0.282,-0.959) []{};
		\vertex[fill] (v010) at (0.541,0.841) []{};
		\draw[thin,fill] (0.837,0.547) circle(0.02);
		\draw[thin,fill] (0.910,0.415) circle(0.02);
		\draw[thin,fill] (0.961,0.275) circle(0.02);
		\draw[thin,fill] (0.845,-0.535) circle(0.02);
		\draw[thin,fill] (0.756,-0.655) circle(0.02);
		\draw[thin,fill] (0.649,-0.760) circle(0.02);

    \filldraw[fill=lightgray,line width=0.34cm,line join=round,draw=black] (2.9590,0.8410) -- (2.5100,-0.1420) -- (2.7345,0.3495) -- cycle;
    \filldraw[fill=lightgray,line width=0.30cm,line join=round,draw=lightgray] (2.9590,0.8410) -- (2.5100,-0.1420) -- (2.7345,0.3495) -- cycle;
    \filldraw[fill=lightgray,line width=0.34cm,line join=round,draw=black] (2.7440,-0.1420) -- (2.5900,0.4150) -- (2.7440,-0.1420) -- (2.7440,-0.6550) -- (3.2180,-0.9590) -- (2.5900,0.4150);
    \filldraw[fill=lightgray,line width=0.30cm,line join=round,draw=lightgray] (2.7440,-0.1420) -- (2.5900,0.4150) -- (2.7440,-0.1420) -- (2.7440,-0.6550) -- (3.2180,-0.9590) -- (2.5900,0.4150);

    \draw[line width=0.02cm,draw=black,fill=lightgray] (3.782,-0.959) circle (0.15cm);
    \draw[line width=0.02cm,draw=black,fill=lightgray] (3.5,1) circle (0.15cm);
    \draw[line width=0.02cm,draw=black,fill=lightgray] (4.041,0.841) circle (0.15cm);

		\vertex[fill] (v10) at (3.5,1) []{};
		\vertex[fill] (v11) at (2.9590,0.8410) []{};
		\vertex[fill] (v12) at (2.5900,0.4150) []{};
		\vertex[fill] (v13) at (2.5100,-0.1420) []{};
		\vertex[fill] (v14) at (2.7440,-0.6550) []{};
		\vertex[fill] (v15) at (3.2180,-0.9590) []{};
		\vertex[fill] (v16) at (3.782,-0.959) []{};
		\vertex[fill] (v110) at (4.041,0.841) []{};
		\draw[thin,fill] (4.337,0.547) circle(0.02);
		\draw[thin,fill] (4.410,0.415) circle(0.02);
		\draw[thin,fill] (4.461,0.275) circle(0.02);
		\draw[thin,fill] (4.345,-0.535) circle(0.02);
		\draw[thin,fill] (4.256,-0.655) circle(0.02);
		\draw[thin,fill] (4.149,-0.760) circle(0.02);

    \filldraw[fill=lightgray,line width=0.34cm,line join=round,draw=black] (6.0900,0.4150) -- (6.0100,-0.1420) -- (6.0500,0.1365) -- cycle;
    \filldraw[fill=lightgray,line width=0.30cm,line join=round,draw=lightgray] (6.0900,0.4150) -- (6.0100,-0.1420) -- (6.0500,0.1365) -- cycle;
    \filldraw[fill=lightgray,line width=0.34cm,line join=round,draw=black] (6.4590,0.8410) -- (6.2440,-0.6550) -- (6.7180,-0.9590) -- cycle;
    \filldraw[fill=lightgray,line width=0.30cm,line join=round,draw=lightgray] (6.4590,0.8410) -- (6.2440,-0.6550) -- (6.7180,-0.9590) -- cycle;

    \draw[line width=0.02cm,draw=black,fill=lightgray] (7.282,-0.959) circle (0.15cm);
    \draw[line width=0.02cm,draw=black,fill=lightgray] (7,1) circle (0.15cm);
    \draw[line width=0.02cm,draw=black,fill=lightgray] (7.541,0.841) circle (0.15cm);

		\vertex[fill] (v20) at (7,1) []{};
		\vertex[fill] (v21) at (6.4590,0.8410) []{};
		\vertex[fill] (v22) at (6.0900,0.4150) []{};
		\vertex[fill] (v23) at (6.0100,-0.1420) []{};
		\vertex[fill] (v24) at (6.2440,-0.6550) []{};
		\vertex[fill] (v25) at (6.7180,-0.9590) []{};
		\vertex[fill] (v26) at (7.282,-0.959) []{};
		\vertex[fill] (v210) at (7.541,0.841) []{};
		\draw[thin,fill] (7.837,0.547) circle(0.02);
		\draw[thin,fill] (7.910,0.415) circle(0.02);
		\draw[thin,fill] (7.961,0.275) circle(0.02);
		\draw[thin,fill] (7.845,-0.535) circle(0.02);
		\draw[thin,fill] (7.756,-0.655) circle(0.02);
		\draw[thin,fill] (7.649,-0.760) circle(0.02);
	\end{tikzpicture}
	\caption{The (essentially unique) partial orthogonality graph that does not violate unextendibility in case 4(a).}\label{fig:case_4a}
\end{figure}

The parties $P_1, P_2, P_3$ can have no more than $(2k+5) + 2(2k+3) = 6k + 11$ distinct edges among them (since there will be a lot of overlap at the left edge of the graphs if we make each group of $3$ equal states orthogonal to the group of $2$ equal states). It is straightforward to see that none of the remaining $4k-3$ parties $P_j$ can have $M_j \geq 3$ or $C_{2,j} \geq 2$ without breaking unextendiblity. Thus those $4k-3$ parties can produce no more than $2k+2$ edges each, for a total of $6k + 11 + (4k-3)(2k+2) = 8k^2 + 8k + 5$ edges. Since $8k^2 + 8k + 5 < 8k^2 + 10k + 3$ when $k \geq 2$, there are some edges missing from the orthogonality graphs, which is a contradiction.

\noindent {\bf Case 4(b):} There exists a party $P_1$ such that $M_1 = 3$, but $M_j \leq 2$ for $j \geq 2$.

Party $P_1$ contributes at most $2k+5$ edges to the orthogonality graph, and the unextendibility requirement implies that $C_{2,j} \leq 3$ for $j \geq 2$. Suppose that there are $m$ indices $2 \leq j_1, j_2, \ldots, j_m \leq 4k$ such that $C_{2,j_i} = 3$ for $1 \leq i \leq m$ and $C_{2,j} = 1$ for all other values of $j$. Then there are at most $(2k+5) + m(2k+4) + (4k-m-1)(2k+2) = 8k^2 + 8k + 2m + 3$ total edges between all $4k$ parties. As in the previous cases, we need a total of $8k^2 + 10k + 3$ edges, which implies that $m \geq k$. We already saw via brute-force search in case~3(b) that we can't have $m \geq 5$, so we only need to rule out the $3 \leq k \leq 4$ cases.

If the group of $3$ identical states on party $P_1$ is represented by vertices $v_3, v_4$, and $v_5$ (see Figure~\ref{fig:case_3c}), then each one of the $3$ groups of $2$ identical states on the other parties must contain exactly one of $v_3, v_4$, or $v_5$. By refining our brute-force computer search to take this restriction into account, we find that there is no configuration of shaded regions that does not violate unextendibility when $m \geq 3$~\cite{JohQubitUPBCode}, so no such UPB exists when $k \geq 3$.

\begin{figure}[htb]
	\centering
	\begin{tikzpicture}[x=1.4cm, y=1.4cm, label distance=0cm]     
    \draw[draw=black] (-0.6760,-0.5853) -- (0.5190,-0.8070);

    \filldraw[fill=lightgray,line width=0.34cm,line join=round,draw=black] (-0.990,-0.142) -- (-0.756,-0.655) -- (-0.282,-0.959) -- cycle;
    \filldraw[fill=lightgray,line width=0.30cm,line join=round,draw=lightgray] (-0.990,-0.142) -- (-0.756,-0.655) -- (-0.282,-0.959) -- cycle;
    \filldraw[fill=lightgray,line width=0.34cm,line join=round,draw=black] (0.282,-0.959) -- (0.756,-0.655) -- (0.5190,-0.8070) -- cycle;
    \filldraw[fill=lightgray,line width=0.30cm,line join=round,draw=lightgray] (0.282,-0.959) -- (0.756,-0.655) -- (0.5190,-0.8070) -- cycle;

    \draw[line width=0.02cm,draw=black,fill=lightgray] (0.990,-0.142) circle (0.15cm);
    \draw[line width=0.02cm,draw=black,fill=lightgray] (-0.541,0.841) circle (0.15cm);
    \draw[line width=0.02cm,draw=black,fill=lightgray] (0,1) circle (0.15cm);
    \draw[line width=0.02cm,draw=black,fill=lightgray] (-0.910,0.415) circle (0.15cm);
    \draw[line width=0.02cm,draw=black,fill=lightgray] (0.910,0.415) circle (0.15cm);
    \draw[line width=0.02cm,draw=black,fill=lightgray] (0.541,0.841) circle (0.15cm);

		\vertex[fill] (v000) at (0,1) [label=90:$v_{0}$]{};
		\vertex[fill] (v001) at (-0.541,0.841) [label=122:$v_{1}$]{};
		\vertex[fill] (v002) at (-0.910,0.415) [label=155:$v_{2}$]{};
		\vertex[fill] (v003) at (-0.990,-0.142) [label=188:$v_{3}$]{};
		\vertex[fill] (v004) at (-0.756,-0.655) [label=221:$v_{4}$]{};
		\vertex[fill] (v005) at (-0.282,-0.959) [label=254:$v_{5}$]{};
		\vertex[fill] (v006) at (0.282,-0.959) [label=286:$v_{6}$]{};
		\vertex[fill] (v007) at (0.756,-0.655) [label=319:$v_{7}$]{};
		\vertex[fill] (v008) at (0.990,-0.142) [label=352:$v_{8}$]{};
		\vertex[fill] (v009) at (0.910,0.415) [label=25:$v_{9}$]{};
		\vertex[fill] (v010) at (0.541,0.841) [label=57:$v_{10}$]{};

    \draw[draw=black] (3.5770,-0.1200) -- (3.8540,-0.5505);

    \filldraw[fill=lightgray,line width=0.34cm,line join=round,draw=black] (2.744,-0.655) -- (4.410,0.415) -- (3.5770,-0.1200) -- cycle;
    \filldraw[fill=lightgray,line width=0.30cm,line join=round,draw=lightgray] (2.744,-0.655) -- (4.410,0.415) -- (3.5770,-0.1200) -- cycle;
    \filldraw[fill=lightgray,line width=0.34cm,line join=round,draw=black] (3.218,-0.959) -- (4.490,-0.142) -- (3.8540,-0.5505) -- cycle;
    \filldraw[fill=lightgray,line width=0.30cm,line join=round,draw=lightgray] (3.218,-0.959) -- (4.490,-0.142) -- (3.8540,-0.5505) -- cycle;

		\vertex[fill] (v100) at (3.5,1) [label=90:$v_{0}$]{};
		\vertex[fill] (v101) at (2.959,0.841) [label=122:$v_{1}$]{};
		\vertex[fill] (v102) at (2.590,0.415) [label=155:$v_{2}$]{};
		\vertex[fill] (v103) at (2.510,-0.142) [label=188:$v_{3}$]{};
		\vertex[fill] (v104) at (2.744,-0.655) [label=221:$v_{4}$]{};
		\vertex[fill] (v105) at (3.218,-0.959) [label=254:$v_{5}$]{};
		\vertex[fill] (v106) at (3.782,-0.959) [label=286:$v_{6}$]{};
		\vertex[fill] (v107) at (4.256,-0.655) [label=319:$v_{7}$]{};
		\vertex[fill] (v108) at (4.490,-0.142) [label=352:$v_{8}$]{};
		\vertex[fill] (v109) at (4.410,0.415) [label=25:$v_{9}$]{};
		\vertex[fill] (v110) at (4.041,0.841) [label=57:$v_{10}$]{};

    \draw[draw=black] (6.010,-0.142) -- (6.244,-0.655);

    \filldraw[fill=lightgray,line width=0.34cm,line join=round,draw=black] (6.010,-0.142) -- (7.990,-0.142) -- (7.0000,-0.1420) -- cycle;
    \filldraw[fill=lightgray,line width=0.30cm,line join=round,draw=lightgray] (6.010,-0.142) -- (7.990,-0.142) -- (7.0000,-0.1420) -- cycle;
    \filldraw[fill=lightgray,line width=0.34cm,line join=round,draw=black] (6.244,-0.655) -- (7.541,0.841) -- (6.8925,0.0930) -- cycle;
    \filldraw[fill=lightgray,line width=0.30cm,line join=round,draw=lightgray] (6.244,-0.655) -- (7.541,0.841) -- (6.8925,0.0930) -- cycle;

		\vertex[fill] (v200) at (7,1) [label=90:$v_{0}$]{};
		\vertex[fill] (v201) at (6.459,0.841) [label=122:$v_{1}$]{};
		\vertex[fill] (v202) at (6.090,0.415) [label=155:$v_{2}$]{};
		\vertex[fill] (v203) at (6.010,-0.142) [label=188:$v_{3}$]{};
		\vertex[fill] (v204) at (6.244,-0.655) [label=221:$v_{4}$]{};
		\vertex[fill] (v205) at (6.718,-0.959) [label=254:$v_{5}$]{};
		\vertex[fill] (v206) at (7.282,-0.959) [label=286:$v_{6}$]{};
		\vertex[fill] (v207) at (7.756,-0.655) [label=319:$v_{7}$]{};
		\vertex[fill] (v208) at (7.990,-0.142) [label=352:$v_{8}$]{};
		\vertex[fill] (v209) at (7.910,0.415) [label=25:$v_{9}$]{};
		\vertex[fill] (v210) at (7.541,0.841) [label=57:$v_{10}$]{};
	\end{tikzpicture}
	\caption{An example of a partial orthogonality graph in case~4(b).}\label{fig:case_3c}
\end{figure}
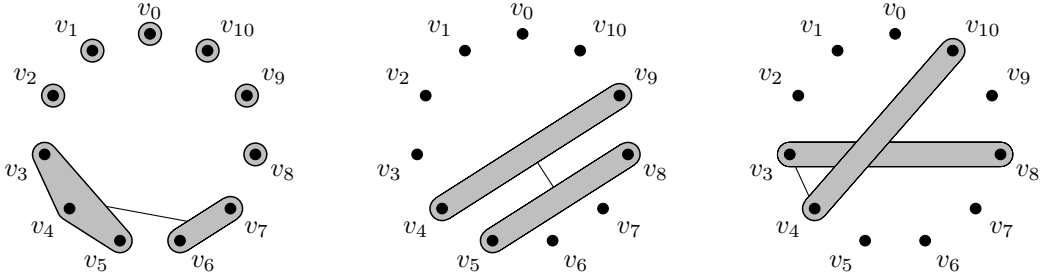

\noindent {\bf Case 4(c):} There exist two parties $P_1, P_2$ such that $M_1 = M_2 = 3$, but $M_j \leq 2$ for $j \geq 3$.

In this case, there are (up to relabelling vertices and parties) only two possible configurations of parties $P_1$ and $P_2$, which are depicted in Figures~\ref{fig:case_4ci} and~\ref{fig:case_4cii}. Notice that in Figure~\ref{fig:case_4ci}, the shaded region on party $P_1$ that contains exactly two vertices \emph{does not} share any common vertices with the shaded region on party $P_2$ that contains exactly two vertices, while in Figure~\ref{fig:case_4cii} those two regions contain the common vertex $v_1$.
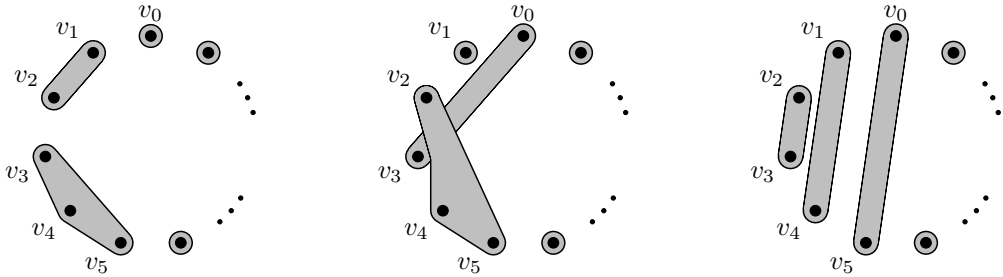
\begin{figure}[htb]
	\centering
	\begin{tikzpicture}[x=1.4cm, y=1.4cm, label distance=0cm]     
    \filldraw[fill=lightgray,line width=0.34cm,line join=round,draw=black] (-0.990,-0.142) -- (-0.756,-0.655) -- (-0.282,-0.959) -- cycle;
    \filldraw[fill=lightgray,line width=0.30cm,line join=round,draw=lightgray] (-0.990,-0.142) -- (-0.756,-0.655) -- (-0.282,-0.959) -- cycle;
    \filldraw[fill=lightgray,line width=0.34cm,line join=round,draw=black] (-0.541,0.841) -- (-0.910,0.415) -- (-0.7255,0.6280) -- cycle;
    \filldraw[fill=lightgray,line width=0.30cm,line join=round,draw=lightgray] (-0.541,0.841) -- (-0.910,0.415) -- (-0.7255,0.6280) -- cycle;

    \draw[line width=0.02cm,draw=black,fill=lightgray] (0.282,-0.959) circle (0.15cm);
    \draw[line width=0.02cm,draw=black,fill=lightgray] (0,1) circle (0.15cm);
    \draw[line width=0.02cm,draw=black,fill=lightgray] (0.541,0.841) circle (0.15cm);

		\vertex[fill] (v00) at (0,1) [label=90:$v_{0}$]{};
		\vertex[fill] (v01) at (-0.541,0.841) [label=122:$v_{1}$]{};
		\vertex[fill] (v02) at (-0.910,0.415) [label=155:$v_{2}$]{};
		\vertex[fill] (v03) at (-0.990,-0.142) [label=188:$v_{3}$]{};
		\vertex[fill] (v04) at (-0.756,-0.655) [label=221:$v_{4}$]{};
		\vertex[fill] (v05) at (-0.282,-0.959) [label=254:$v_{5}$]{};
		\vertex[fill] (v06) at (0.282,-0.959) []{};
		\vertex[fill] (v010) at (0.541,0.841) []{};
		\draw[thin,fill] (0.837,0.547) circle(0.02);
		\draw[thin,fill] (0.910,0.415) circle(0.02);
		\draw[thin,fill] (0.961,0.275) circle(0.02);
		\draw[thin,fill] (0.845,-0.535) circle(0.02);
		\draw[thin,fill] (0.756,-0.655) circle(0.02);
		\draw[thin,fill] (0.649,-0.760) circle(0.02);

    \filldraw[fill=lightgray,line width=0.34cm,line join=round,draw=black] (3.5,1) -- (2.5100,-0.1420) -- (3.0050,0.4290) -- cycle;
    \filldraw[fill=lightgray,line width=0.30cm,line join=round,draw=lightgray] (3.5,1) -- (2.5100,-0.1420) -- (3.0050,0.4290) -- cycle;
    \filldraw[fill=lightgray,line width=0.34cm,line join=round,draw=black] (2.7440,-0.1420) -- (2.5900,0.4150) -- (2.7440,-0.1420) -- (2.7440,-0.6550) -- (3.2180,-0.9590) -- (2.5900,0.4150);
    \filldraw[fill=lightgray,line width=0.30cm,line join=round,draw=lightgray] (2.7440,-0.1420) -- (2.5900,0.4150) -- (2.7440,-0.1420) -- (2.7440,-0.6550) -- (3.2180,-0.9590) -- (2.5900,0.4150);

    \draw[line width=0.02cm,draw=black,fill=lightgray] (3.782,-0.959) circle (0.15cm);
    \draw[line width=0.02cm,draw=black,fill=lightgray] (2.9590,0.8410) circle (0.15cm);
    \draw[line width=0.02cm,draw=black,fill=lightgray] (4.041,0.841) circle (0.15cm);

		\vertex[fill] (v10) at (3.5,1) [label=90:$v_{0}$]{};
		\vertex[fill] (v11) at (2.9590,0.8410) [label=122:$v_{1}$]{};
		\vertex[fill] (v12) at (2.5900,0.4150) [label=155:$v_{2}$]{};
		\vertex[fill] (v13) at (2.5100,-0.1420) [label=188:$v_{3}$]{};
		\vertex[fill] (v14) at (2.7440,-0.6550) [label=221:$v_{4}$]{};
		\vertex[fill] (v15) at (3.2180,-0.9590) [label=254:$v_{5}$]{};
		\vertex[fill] (v16) at (3.782,-0.959) []{};
		\vertex[fill] (v110) at (4.041,0.841) []{};
		\draw[thin,fill] (4.337,0.547) circle(0.02);
		\draw[thin,fill] (4.410,0.415) circle(0.02);
		\draw[thin,fill] (4.461,0.275) circle(0.02);
		\draw[thin,fill] (4.345,-0.535) circle(0.02);
		\draw[thin,fill] (4.256,-0.655) circle(0.02);
		\draw[thin,fill] (4.149,-0.760) circle(0.02);

    \filldraw[fill=lightgray,line width=0.34cm,line join=round,draw=black] (6.0900,0.4150) -- (6.0100,-0.1420) -- (6.0500,0.1365) -- cycle;
    \filldraw[fill=lightgray,line width=0.30cm,line join=round,draw=lightgray] (6.0900,0.4150) -- (6.0100,-0.1420) -- (6.0500,0.1365) -- cycle;
    \filldraw[fill=lightgray,line width=0.34cm,line join=round,draw=black] (6.4590,0.8410) -- (6.2440,-0.6550) -- (6.3515,0.0930) -- cycle;
    \filldraw[fill=lightgray,line width=0.30cm,line join=round,draw=lightgray] (6.4590,0.8410) -- (6.2440,-0.6550) -- (6.3515,0.0930) -- cycle;
	\filldraw[fill=lightgray,line width=0.34cm,line join=round,draw=black] (7,1) -- (6.7180,-0.9590) -- (6.8590,0.0205) -- cycle;
    \filldraw[fill=lightgray,line width=0.30cm,line join=round,draw=lightgray] (7,1) -- (6.7180,-0.9590) -- (6.8590,0.0205) -- cycle;

    \draw[line width=0.02cm,draw=black,fill=lightgray] (7.282,-0.959) circle (0.15cm);
    \draw[line width=0.02cm,draw=black,fill=lightgray] (7.541,0.841) circle (0.15cm);

		\vertex[fill] (v20) at (7,1) [label=90:$v_{0}$]{};
		\vertex[fill] (v21) at (6.4590,0.8410) [label=122:$v_{1}$]{};
		\vertex[fill] (v22) at (6.0900,0.4150) [label=155:$v_{2}$]{};
		\vertex[fill] (v23) at (6.0100,-0.1420) [label=188:$v_{3}$]{};
		\vertex[fill] (v24) at (6.2440,-0.6550) [label=221:$v_{4}$]{};
		\vertex[fill] (v25) at (6.7180,-0.9590) [label=254:$v_{5}$]{};
		\vertex[fill] (v26) at (7.282,-0.959) []{};
		\vertex[fill] (v210) at (7.541,0.841) []{};
		\draw[thin,fill] (7.837,0.547) circle(0.02);
		\draw[thin,fill] (7.910,0.415) circle(0.02);
		\draw[thin,fill] (7.961,0.275) circle(0.02);
		\draw[thin,fill] (7.845,-0.535) circle(0.02);
		\draw[thin,fill] (7.756,-0.655) circle(0.02);
		\draw[thin,fill] (7.649,-0.760) circle(0.02);
	\end{tikzpicture}
	\caption{One of two possible partial orthogonality graphs of parties $P_1$, $P_2$, and $P_3$ that does not violate unextendibility in case 4(c).}\label{fig:case_4ci}
\end{figure}

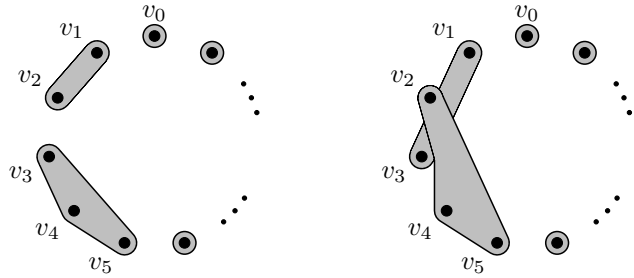
\begin{figure}[htb]
	\centering
	\begin{tikzpicture}[x=1.4cm, y=1.4cm, label distance=0cm]     
    \filldraw[fill=lightgray,line width=0.34cm,line join=round,draw=black] (-0.990,-0.142) -- (-0.756,-0.655) -- (-0.282,-0.959) -- cycle;
    \filldraw[fill=lightgray,line width=0.30cm,line join=round,draw=lightgray] (-0.990,-0.142) -- (-0.756,-0.655) -- (-0.282,-0.959) -- cycle;
    \filldraw[fill=lightgray,line width=0.34cm,line join=round,draw=black] (-0.541,0.841) -- (-0.910,0.415) -- (-0.7255,0.6280) -- cycle;
    \filldraw[fill=lightgray,line width=0.30cm,line join=round,draw=lightgray] (-0.541,0.841) -- (-0.910,0.415) -- (-0.7255,0.6280) -- cycle;

    \draw[line width=0.02cm,draw=black,fill=lightgray] (0.541,0.841) circle (0.15cm);
    \draw[line width=0.02cm,draw=black,fill=lightgray] (0,1) circle (0.15cm);
    \draw[line width=0.02cm,draw=black,fill=lightgray] (0.282,-0.959) circle (0.15cm);

		\vertex[fill] (v00) at (0,1) [label=90:$v_{0}$]{};
		\vertex[fill] (v01) at (-0.541,0.841) [label=122:$v_{1}$]{};
		\vertex[fill] (v02) at (-0.910,0.415) [label=155:$v_{2}$]{};
		\vertex[fill] (v03) at (-0.990,-0.142) [label=188:$v_{3}$]{};
		\vertex[fill] (v04) at (-0.756,-0.655) [label=221:$v_{4}$]{};
		\vertex[fill] (v05) at (-0.282,-0.959) [label=254:$v_{5}$]{};
		\vertex[fill] (v06) at (0.282,-0.959) []{};
		\draw[thin,fill] (0.837,0.547) circle(0.02);
		\draw[thin,fill] (0.910,0.415) circle(0.02);
		\draw[thin,fill] (0.961,0.275) circle(0.02);
		\draw[thin,fill] (0.845,-0.535) circle(0.02);
		\draw[thin,fill] (0.756,-0.655) circle(0.02);
		\draw[thin,fill] (0.649,-0.760) circle(0.02);
		\vertex[fill] (v010) at (0.541,0.841) []{};

    \filldraw[fill=lightgray,line width=0.34cm,line join=round,draw=black] (2.9590,0.8410) -- (2.5100,-0.1420) -- (2.7345,0.3495) -- cycle;
    \filldraw[fill=lightgray,line width=0.30cm,line join=round,draw=lightgray] (2.9590,0.8410) -- (2.5100,-0.1420) -- (2.7345,0.3495) -- cycle;
    \filldraw[fill=lightgray,line width=0.34cm,line join=round,draw=black] (2.7440,-0.1420) -- (2.5900,0.4150) -- (2.7440,-0.1420) -- (2.7440,-0.6550) -- (3.2180,-0.9590) -- (2.5900,0.4150);
    \filldraw[fill=lightgray,line width=0.30cm,line join=round,draw=lightgray] (2.7440,-0.1420) -- (2.5900,0.4150) -- (2.7440,-0.1420) -- (2.7440,-0.6550) -- (3.2180,-0.9590) -- (2.5900,0.4150);

    \draw[line width=0.02cm,draw=black,fill=lightgray] (3.782,-0.959) circle (0.15cm);
    \draw[line width=0.02cm,draw=black,fill=lightgray] (3.5,1) circle (0.15cm);
    \draw[line width=0.02cm,draw=black,fill=lightgray] (4.041,0.841) circle (0.15cm);

		\vertex[fill] (v10) at (3.5,1) [label=90:$v_{0}$]{};
		\vertex[fill] (v11) at (2.9590,0.8410) [label=122:$v_{1}$]{};
		\vertex[fill] (v12) at (2.5900,0.4150) [label=155:$v_{2}$]{};
		\vertex[fill] (v13) at (2.5100,-0.1420) [label=188:$v_{3}$]{};
		\vertex[fill] (v14) at (2.7440,-0.6550) [label=221:$v_{4}$]{};
		\vertex[fill] (v15) at (3.2180,-0.9590) [label=254:$v_{5}$]{};
		\vertex[fill] (v16) at (3.782,-0.959) []{};
		\draw[thin,fill] (4.337,0.547) circle(0.02);
		\draw[thin,fill] (4.410,0.415) circle(0.02);
		\draw[thin,fill] (4.461,0.275) circle(0.02);
		\draw[thin,fill] (4.345,-0.535) circle(0.02);
		\draw[thin,fill] (4.256,-0.655) circle(0.02);
		\draw[thin,fill] (4.149,-0.760) circle(0.02);
		\vertex[fill] (v110) at (4.041,0.841) []{};
	\end{tikzpicture}
	\caption{The other possible partial orthogonality graph of parties $P_1, P_2$ that does not violate unextendibility in case 4(c).}\label{fig:case_4cii}
\end{figure}

Suppose for now that parties $P_1$ and $P_2$ have a total of at most $4k+8$ distinct edges on their orthogonality graphs. If there are $m$ parties $P_j$ ($j \geq 3$) for which $C_{2,j} = 3$, then we have a total of at most $(4k+8) + m(2k+4) + (4k-m-2)(2k+2) = 8k^2 + 8k + 2m + 4$ edges. Any all of these $m$ parties, we require that one of the shaded regions contains $v_2$ and $v_3$ and the other shaded regions containing two vertices each contain one of $v_4$ or $v_5$. Thus, the brute-force search described in case~4(b) applies here as well and shows that $m \leq 2$. However, when $m = 2$ we have $8k^2 + 8k + 2m + 4 = 8k^2 + 8k + 8 < 8k^2 + 10k + 3$ when $k \geq 3$, which shows that there can not possibly be enough edges on the orthogonality graphs in this case.

The only remaining possibility is that the parties $P_1$ and $P_2$ have a total of at least $4k+9$ distinct edges (and hence \emph{exactly} $4k+9$ distinct edges). In this case, parties $P_1$ and $P_2$ must be as in Figure~\ref{fig:case_4ci}, and on both of the parties $P_1$ and $P_2$ the set of $3$ equal states must be orthogonal to the set of $2$ equal states. Furthermore, it is not difficult to show that in this case, any party $P_j$ with $C_{2,j} = 3$ can introduce at most $2k + 3$ new edges that are not already present in the orthogonality graph of parties $P_1$ and $P_2$. Thus, if there are $m$ parties $P_j$ ($j \geq 3$) for which $C_{2,j} = 3$, we have a total of at most $(4k+9) + m(2k+3) + (4k-m-2)(2k+2) = 8k^2 + 8k + m + 5$ edges. Since $m \leq 2$ (as before) and $k \geq 3$, it follows that $8k^2 + 8k + m + 5 < 8k^2 + 10k + 3$, which again shows that there can not possibly be enough edges on the orthogonality graphs in this case.
\end{proof}

\subparagraph*{Acknowledgements}

Thanks are extended to Gus Gutoski for suggesting a computer search to fill in the gaps in the proof of Lemma~\ref{lem:upb_p3}. The author was supported by the Natural Sciences and Engineering Research Council of Canada and the Mprime Network.

\bibliography{../../_bibliographies_/quantum}
\end{document}